\documentclass[11pt]{article}
\usepackage[utf8]{inputenc}
\usepackage[margin=1in]{geometry}
\usepackage[T1]{fontenc} 
\usepackage{graphicx}
\usepackage{mathpazo}
\usepackage{booktabs} 
\usepackage{multirow} 
\usepackage{multicol} 
\usepackage{color-edits} 
\addauthor[Grigoris]{gv}{teal}

\usepackage[dvipsnames]{xcolor}
\definecolor{niceRed}{RGB}{190,38,38}
\definecolor{blueGrotto}{HTML}{059DC0}
\definecolor{royalBlue}{HTML}{057DCD}
\definecolor{navyBlue}{HTML}{0B579C}
\definecolor{limeGreen}{HTML}{81B622}
\definecolor{nicePurple}{HTML}{9c27b0}
\definecolor{lightRoyalBlue}{HTML}{def2ff} 
 \definecolor{gold}{HTML}{ffa300}

\usepackage{natbib}

\usepackage[utf8]{inputenc}

\usepackage{amsmath}
\usepackage{amssymb}
\usepackage{amsthm}
\usepackage{bbm}
\usepackage{blkarray}
\usepackage{color}
\usepackage{enumerate}
\usepackage{float} 
\usepackage{subcaption}
\usepackage{tikz}
\usetikzlibrary{arrows.meta, positioning, calc, shapes, fit}
\usepackage{hyperref}
\usepackage[capitalise,noabbrev,nameinlink]{cleveref} 
\usepackage{xcolor}
\definecolor{DrexelBlue}{HTML}{003058}
\definecolor{commentblue}{HTML}{2A6EC7}

\usepackage{hyperref}
\hypersetup{
  colorlinks = true,  
  urlcolor = {commentblue},
  linkcolor = {DrexelBlue}, 
  citecolor = {YellowOrange!85!black}
}
\usepackage{mathtools}
\usepackage[inline]{enumitem} 
\usepackage{tcolorbox}
\usepackage{nicefrac}
\usepackage{dsfont}
\usepackage[framemethod=TikZ]{mdframed}
\usepackage{todonotes}
\usepackage{changepage} 
\usepackage{pifont} 
\usepackage{thm-restate} 
\usepackage{centernot}

\tikzset{
  myNodeFlex/.style={
    draw,
    rectangle,
    rounded corners,
    text centered,
    minimum height=1.5em,
  }
}

\tikzset{
  myNode/.style={
    draw,
    rectangle,
    rounded corners,
    text centered,
    minimum height=1.5em,
    minimum width=3cm,
    text width=5cm,    
  }
}

\tikzset{
  myNodeNarrow/.style={
    draw,
    rectangle,
    rounded corners,
    text centered,
    minimum height=1.5em,
    minimum width=1cm,
  }
}

\tikzset{
  myNodeWide/.style={
    draw,
    rectangle,
    rounded corners,
    text centered,
    minimum height=1.5em,
    minimum width=6cm,
  }
}

\usepackage[linesnumbered,noend,ruled,noline]{algorithm2e} 
\SetAlFnt{\small}
\SetAlCapFnt{\small}
\SetAlCapNameFnt{\small}
\SetAlCapHSkip{0pt}
\IncMargin{-\parindent}
\SetArgSty{textnormal}
 
\SetCommentSty{mycommfont}
 \usepackage{enumitem}

\usepackage{soul} 

\theoremstyle{plain} 
\newtheorem{theorem}{Theorem}[section]
\newtheorem{corollary}[theorem]{Corollary}
\newtheorem{proposition}[theorem]{Proposition}
\newtheorem{lemma}[theorem]{Lemma}

\newtheorem{definition}{Definition}

\newtheorem*{definition*}{Definition}

\theoremstyle{definition} 

\theoremstyle{remark} 

\newtheorem{remark}{Remark}

\AfterEndEnvironment{definition}{\noindent\ignorespaces}
\AfterEndEnvironment{infdefinition}{\noindent\ignorespaces}
\AfterEndEnvironment{example}{\noindent\ignorespaces}
\AfterEndEnvironment{assumption}{\noindent\ignorespaces}
\AfterEndEnvironment{lemma}{\noindent\ignorespaces}
\AfterEndEnvironment{theorem}{\noindent\ignorespaces}
\AfterEndEnvironment{proposition}{\noindent\ignorespaces}
\AfterEndEnvironment{fact}{\noindent\ignorespaces}
\AfterEndEnvironment{question}{\noindent\ignorespaces}
\AfterEndEnvironment{corollary}{\noindent\ignorespaces}
\AfterEndEnvironment{model}{\noindent\ignorespaces}
\AfterEndEnvironment{remark}{\noindent\ignorespaces}
\AfterEndEnvironment{proof}{\noindent\ignorespaces}
\AfterEndEnvironment{fact}{\noindent\ignorespaces}
\AfterEndEnvironment{minftheorem}{\noindent\ignorespaces}
\AfterEndEnvironment{inftheorem}{\noindent\ignorespaces}
\AfterEndEnvironment{maintheorem}{\noindent\ignorespaces}
\AfterEndEnvironment{restatable}{\noindent\ignorespaces}
\AfterEndEnvironment{observation}{\noindent\ignorespaces}

\crefname{section}{Section}{Sections}
\crefname{theorem}{Theorem}{Theorems}
\crefname{theorem*}{Theorem}{Theorems}
\crefname{inftheorem}{Informal Theorem}{Informal Theorems}
\crefname{assumption}{Assumption}{Assumptions}
\crefname{lemma}{Lemma}{Lemmas}
\crefname{definition}{Definition}{Definitions}
\crefname{infdefinition}{Informal Definition}{Informal Definitions}
\crefname{conjecture}{Conjecture}{Conjectures}
\crefname{corollary}{Corollary}{Corollaries}
\crefname{construction}{Construction}{Constructions}
\crefname{conjecture}{Conjecture}{Conjectures}
\crefname{claim}{Claim}{Claims}
\crefname{observation}{Observation}{Observations}
\crefname{proposition}{Proposition}{Propositions}
\crefname{fact}{Fact}{Facts}
\crefname{question}{Question}{Questions}
\crefname{problem}{Problem}{Problems}
\crefname{remark}{Remark}{Remarks}
\crefname{example}{Example}{Examples}
\crefname{equation}{Equation}{Equations}
\crefname{appendix}{Appendix}{Appendices}
\crefname{algorithm}{Algorithm}{Algorithms}
\crefname{model}{Model}{Models}
\crefname{figure}{Figure}{Figures}
\crefname{condition}{Condition}{Conditions}
\crefname{algocf}{algorithm}{algorithms}
\Crefname{algocf}{Algorithm}{Algorithms}

\usepackage{etoolbox}

\BeforeBeginEnvironment{subenvironment}{\white{.}\\ \vspace{-10mm}}
\AfterEndEnvironment{subenvironment}{\vspace{4mm}}

\newcommand{\eat}[1]{}

\makeatletter
\newcommand{\tagnum}[2]{%
    \refstepcounter{equation}%
    \tag{#1) \ (\theequation}%
    \protected@write \@auxout {}{%
        \string \newlabel {#2}{{\theequation}{\thepage}{}{equation.\theequation}{}}%
    }%
}
\makeatother

\newcommand{\white}[1]{\textcolor{white}{#1}}

\newcommand{\N}{\mathbb{N}}

\newcommand{\poly}{\mathrm{poly}}

\newcommand{\eps}{\varepsilon}
\renewcommand{\epsilon}{\varepsilon}

\makeatletter
\newcommand*{\tran}{{\mathpalette\@tran{}}}
\newcommand*{\@tran}[2]{\raisebox{\depth}{$\m@th#1\intercal$}}
\makeatother

\renewcommand{\tilde}{\widetilde}

\def\<{\langle}
\def\>{\rangle}

\DeclareMathAlphabet{\mathpzc}{OT1}{pzc}{m}{it}

\newcommand{\customcal}[1]{\euscr{#1}}

\newcommand{\cC}{\customcal{C}}

\newcommand{\cS}{\customcal{S}}
\newcommand{\cT}{\customcal{T}}

\DeclareMathAlphabet{\mathdutchcal}{U}{dutchcal}{m}{n}
\SetMathAlphabet{\mathdutchcal}{bold}{U}{dutchcal}{b}{n}
\DeclareMathAlphabet{\mathdutchbcal}{U}{dutchcal}{b}{n}
 
\DeclareMathAlphabet\urwscr{U}{urwchancal}{b}{n}%
\DeclareMathAlphabet\rsfscr{U}{rsfso}{m}{n}
\DeclareMathAlphabet\euscr{U}{eus}{m}{n}
\DeclareFontEncoding{LS2}{}{}
\DeclareFontSubstitution{LS2}{stix}{m}{n}
\DeclareMathAlphabet\stixcal{LS2}{stixcal}{m} {n}

\renewcommand{\paragraph}[1]{\bigskip \noindent\textbf{#1}~~}

        \usepackage{tikz}
\usetikzlibrary{patterns}

\renewcommand{\hat}{\widehat}
\renewcommand{\tilde}{\widetilde}

\DeclareMathOperator{\KL}{D_{KL}}

\newcommand{\TV}{d_{\mathrm{TV}}}
\newcommand{\Prob}{\mathbb{P}}
\newcommand{\1}{\mathbf{1}}

\newcommand{\COND}{\mathsf{COND}}

\newcommand{\unifzero}{\mathsf{UNIF0}}
\newcommand{\fail}{\mathsf{FAIL}}
\newcommand{\bits}{\{-1,+1\}}

\newcommand{\co}{\texttt{CO}}
\newcommand{\sellers}{\mathcal{S}}

\usepackage{array}
\newcolumntype{L}[1]{>{\raggedright\let\newline\\\arraybackslash\hspace{0pt}}m{#1}}
\newcolumntype{C}[1]{>{\centering\let\newline\\\arraybackslash\hspace{0pt}}m{#1}}
\newcolumntype{R}[1]{>{\raggedleft\let\newline\\\arraybackslash\hspace{0pt}}m{#1}}

\newmdenv[
    backgroundcolor=lightgray!10, 
    roundcorner=5pt,            
    linecolor=black,             
    linewidth=1pt,               
    innertopmargin=0pt,         
    innerbottommargin=0pt,      
    innerleftmargin=5pt,        
    innerrightmargin=5pt,       
    skipabove=2pt,              
    skipbelow=0pt               
]{curvybox}

\usepackage{newpxtext}
\usepackage{newpxmath}

\title{Learning Distributions from Multiple Data Providers}

\author{%
\begin{tabular}{cc}
\begin{minipage}[t]{0.4\textwidth}\centering
Jon Kleinberg\\
\small Cornell University\\
\small\href{mailto:kleinberg@cornell.edu}{kleinberg@cornell.edu}
\end{minipage}
&
\begin{minipage}[t]{0.4\textwidth}\centering
Amin Saberi\\
\small Stanford University\\
\small\href{mailto:saberi@stanford.edu}{saberi@stanford.edu}
\end{minipage}
\\[8ex]
\begin{minipage}[t]{0.4\textwidth}\centering
Xizhi Tan\\
\small Stanford University\\
\small\href{mailto:xizhi@stanford.edu}{xizhi@stanford.edu}
\end{minipage}
&
\begin{minipage}[t]{0.4\textwidth}\centering
Grigoris Velegkas\\
\small Google Research\\
\small\href{mailto:gvelegkas@google.com}{gvelegkas@google.com}
\end{minipage}
\end{tabular}%
}
\date{}

\begin{document}

\maketitle

\begin{abstract}
Motivated by learning from heterogeneous and overlapping data providers, we study a stylized model of distribution learning from
restricted conditional samples. The goal is to learn an unknown distribution
$p$ on a finite domain $[n]$. The learner is given a fixed family of queryable sets
$\mathcal S \subseteq 2^{[n]}$, and each query to
$S \in \mathcal S$ returns an independent sample from the conditional
distribution $p(\cdot \mid S)$.

Learnability is governed by the \emph{co-occurrence graph} associated with
$\mathcal S$: two domain elements are adjacent if they appear together in some
queryable set. Pointwise consistency is achievable  when this graph is
connected on the target support. PAC learning requires more: it is possible
 when the co-occurrence graph is complete.

The optimal sample complexity of PAC learning ranges from nearly linear to quadratic.
Every query family with complete co-occurrence graph admits sample complexity
$\widetilde O(n^2/\epsilon^2)$, and this bound is tight in the worst case. On
the other hand, if $[n] \in \sellers$ ordinary sampling improves the bound to $\Theta(n/\epsilon^2)$, and this cannot be improved further even if every set is queryable. More generally, we identify
\emph{hierarchical comparability} as a sufficient structural condition on
$\mathcal S$ under which the optimal complexity is nearly linear,
$\widetilde \Theta(n/\epsilon^2)$, with pairwise query families as a canonical
example. Finally, the full range of polynomial rates between linear and
quadratic is attainable: for every $\alpha \in (1,2)$, there exists a query
family with optimal PAC rate $\widetilde \Theta(n^\alpha/\epsilon^2)$.
\end{abstract}

\section{Introduction}

Modern generative models have benefited enormously from broad public and
web-scale data,
but high-quality, human-generated data may
become a limiting resource for continued model scaling \citep{VillalobosEtAl2024}.
As model developers look beyond broad web crawling, curated and provider-controlled data sources become increasingly important. The central question is not simply how much data to acquire, but from which
sources and how their coverage patterns affect learning.

This challenge stems from two inherent phenomena: \emph{coverage} and \emph{redundancy}. A provider specialized in a narrow domain offers deep localized coverage, but does not reveal much about how that domain should be weighted within the global distribution. Conversely, querying multiple broad providers often introduces wasteful redundancy if they have a lot of overlap.

To isolate these issues, we study perhaps the simplest possible version of
generative modeling: learning an unknown discrete distribution $p$ over a finite
domain $[n]$. The elements of $[n]$ may be viewed as document clusters, or more
generally as latent content types. A data provider is represented by a subset
$S \subseteq [n]$. When queried, the provider returns an independent sample from
the conditional distribution $p(\cdot \mid S)$. The learner has access only to a
fixed collection of such subsets, denoted by $\mathcal S$, and its goal is to
output a distribution close to $p$ in total variation distance.

Our framework raises two fundamental questions. The first is qualitative: what
overlap structure is necessary and sufficient for learning, either pointwise or
in the PAC sense? The second is quantitative: once learning is possible, what is
the optimal sample complexity? In particular, does the same structure that
determines learnability also determine the rate of learning, or can different
query families with the same qualitative learnability behavior have different
sample complexities? Since our focus is on high-dimensional domains, we track
the dependence on both the domain size $n$ and the target accuracy $\epsilon$.

This is a restricted variant of the conditional-sampling model, commonly
denoted by $\COND$. Prior work on $\COND$ 
typically assumes unrestricted, or otherwise well-structured, access to
conditional samples, and uses this access to accelerate distribution testing,
frequency estimation, and related tasks~\citep{CanonneRonServedio2015,ChakrabortyFischerGoldhirshMatsliah2016,
Canonne2020Survey}.
In contrast, because we are motivated by settings in which data providers might offer sources that overlap in complex and arbitrary ways, our goal is to understand the learning landscape induced by an arbitrary fixed query family $\mathcal S$, with a focus on learning in total variation distance.

\subsection{Our Results}\label{sec:ourresult}

Our results give a structural theory of distribution learning in total variation
distance from an arbitrary query family $\mathcal S$. They divide into two
parts: qualitative learnability, which identifies when recovery is possible, and
quantitative sample complexity, which determines how the geometry of
$\mathcal S$ controls the optimal rate.

\paragraph{Characterization of learnability.}
The qualitative results, summarized in \cref{tab:market_capacity}, give exact
combinatorial characterizations of PAC learnability and pointwise consistency.
The central object is the \emph{co-occurrence graph}\footnote{In the hypergraph
theory literature, this object appears under several names, including
2-section, clique graph, representing graph, primal graph, and Gaifman
graph, see, e.g.,\citep{Bretto2013Hypergraph}.} induced by the query family $\mathcal S$ on a
target support $U$; see \cref{def:cograph}. Its vertices are the elements of
$U$, and two vertices $x,y \in U$ are adjacent exactly when some queryable set
$S \in \mathcal S$ contains both. This graph is denoted by
$\co(\mathcal S,U)$.

\begin{table}[h]
\centering
\resizebox{0.95\linewidth}{!}{
\renewcommand{\arraystretch}{1.5}
\begin{tabular}{@{}p{6cm}p{9.5cm}@{}}
\toprule
\textbf{Learning guarantee} & \textbf{Necessary and sufficient condition for target support $U$} \\ \midrule
{PAC Learnability of $\mathcal P_U$} & $U$ induces a \textbf{complete} (clique) $\co(\sellers, U)$ \\ 
{Pointwise Consistency on $\mathcal P_U^+$} & $U$ induces a \textbf{connected} $\co(\sellers, U)$ \\ 
{Pointwise Consistency on $\mathcal P_U$} & $U$ induces a \textbf{complete} $\co(\sellers, U)$ \\ 
\bottomrule
\end{tabular}
}
\caption{ Learnability characterization of a query family $\sellers$. Given $U \subseteq [n]$,
$\mathcal P_U^+$ denotes the set of distributions supported on $U$ and $\mathcal P_U$ the set of distributions supported on any subset of $U$.}
\label{tab:market_capacity}
\end{table}

We first consider PAC learnability (\cref{def:uniform}), which requires that for any error target $\epsilon$ there is a single sample complexity bound that holds simultaneously over all distributions in the target class. We show that this is possible if and only if the co-occurrence graph on $U$ is complete (\cref{thm:market_pac}). (Note that a complete co-occurrence graph is consistent with a wide range of possible underlying hypergraph structures.)
 We then consider pointwise consistency (\cref{def:universal}), where the learning rate is allowed to be distribution-dependent. Here, the requirement on $\co(\sellers,U)$ is significantly weaker: assuming the learner knows the exact support of the distribution (the target class is $\mathcal{P}_U^+$), connectivity of the co-occurrence graph is both necessary and sufficient (\cref{thm:market_pointwise}).\footnote{If the exact support is unknown and allows for zero-mass elements (the class $\mathcal{P}_U$), the potential for zero-mass ``bridge'' points breaks this connectivity guarantee, and pointwise consistency reverts to requiring a complete graph.}

 \paragraph{Optimal sample complexity.} 
Having established when learning is possible, the next question is the optimal
sample complexity of PAC learning. The resulting landscape is summarized in
\cref{fig:complexity_landscape}. As an ideal benchmark, with maximal query power
(that is when $\sellers$ consists of all subsets of $[n]$) the optimal bound
is $\Theta(n/\epsilon^2)$ (\cref{thm:full-cond-main}).\footnote{While this
upper bound can be trivially achieved by repeatedly querying the full domain
$[n]$, such direct access may be precluded by
$\sellers$.}

Completeness of the co-occurrence graph characterizes PAC learnability, but it
does not determine the optimal rate. In contrast to classical settings such as
binary classification, where the qualitative condition for learnability
(finite VC dimension) also pins down the quantitative sample complexity
\citep{shalev2014understanding}, complete query families can exhibit much
larger complexity. In particular, there exist families $\sellers$ that induce
complete co-occurrence graphs for which any PAC learner requires
$\Omega(n^2/\epsilon^2)$ samples (\cref{thm:quadratic-lower-main}). Thus,
relative to full conditional access, restricted query families can
incur an additional factor of $n$, even when PAC learning is possible. This is
the true worst-case degradation: completeness alone guarantees a general
$\tilde{O}(n^2/\epsilon^2)$ sample complexity bound
(\cref{thm:generic-upper-main}).

The gap between the nearly linear benchmark and the quadratic worst case raises
a natural structural question: which properties of $\sellers$ allow for better
bounds? Favorable local structure can lead to nearly linear sample complexity.
In particular, \emph{hierarchical comparability} (\cref{def:HC}) is a
sufficient condition under which the optimal complexity is
$\tilde{\Theta}(n/\epsilon^2)$. This condition requires that the domain admit a
recursive, tree-like partition in which internal nodes have appropriate
``local queries''. 
The pure pairwise-query setting, closely related to PCOND and to
Bradley--Terry--Luce comparisons \citep{BradleyTerry1952}, where \(S\) consists of every pair
\(\{x,y\}\), serves as a canonical example.

Finally, linear and quadratic dependence on $n$ are not the only possible
rates. For every $\gamma \in (0,1)$, there exist families $\sellers$ that induce
complete co-occurrence graphs and have optimal sample complexity
$\tilde{\Theta}(n^{1+\gamma}/\epsilon^2)$
(\cref{thm:intermediate-rates}).

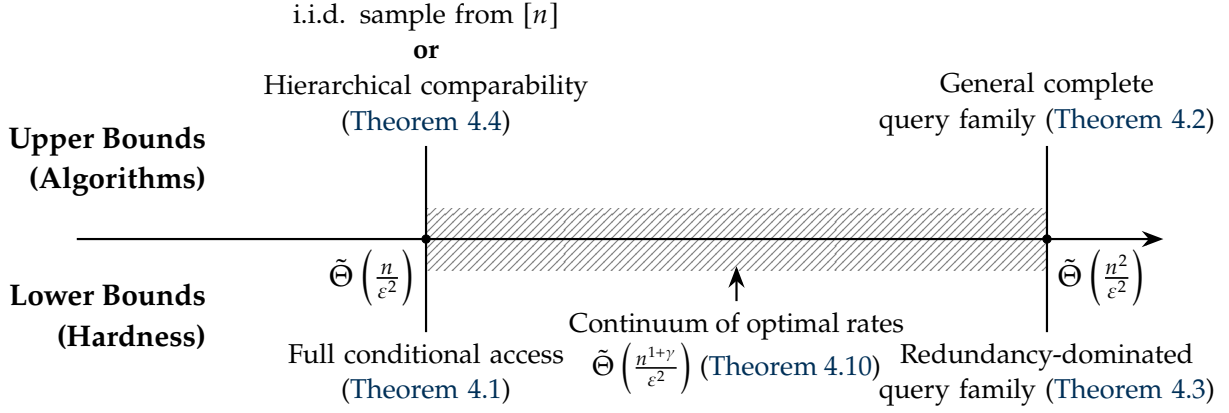
\begin{figure}[h]
\centering
\resizebox{\linewidth}{!}{
\begin{tikzpicture}[
    >={Stealth[scale=1.2]},
    font=\small,
    label node/.style={align=center, text width=4.5cm}
]

\def\xlin{2.5}
\def\xquad{10.5}

\fill[pattern=north east lines, pattern color=gray!100] (\xlin, -0.4) rectangle (\xquad, 0.4);

\draw[->, thick] (-2, 0) -- (12, 0);

\node[anchor=east, align=right, font=\bfseries] at (-0.2, 1) {Upper Bounds\\(Algorithms)};
\node[anchor=east, align=right, font=\bfseries] at (-0.2, -1) {Lower Bounds\\(Hardness)};

\draw[thick] (\xlin, -1.2) -- (\xlin, 1.2);
\filldraw[black] (\xlin, 0) circle (1.5pt) node[below left, font=\normalsize] {$\tilde{\Theta}\left(\frac{n}{\eps^2}\right)$};
\node[anchor=south, label node] at (\xlin, 1.2) {i.i.d. sample from $[n]$\\ \textbf{or}\\ Hierarchical comparability \\
(\cref{thm:tree-local-main})};
\node[anchor=north, label node] at (\xlin, -1.2) {Full conditional access \\ (\cref{thm:full-cond-main})};

\draw[thick] (\xquad, -1.2) -- (\xquad, 1.2);
\filldraw[black] (\xquad, 0) circle (1.5pt) node[below right, font=\normalsize] {$\tilde{\Theta}\left(\frac{n^2}{\eps^2}\right)$};
\node[anchor=south, label node] at (\xquad, 1.2) {General complete \\ query family
(\cref{thm:generic-upper-main})};
\node[anchor=north, label node] at (\xquad, -1.2) {Redundancy-dominated\\query family (\cref{thm:quadratic-lower-main})};

\node[anchor=north, align=center] (poly) at (6.5, -0.8) {Continuum of optimal rates\\$\tilde{\Theta}\left(\frac{n^{1+\gamma}}{\eps^2}\right)$ (\cref{thm:intermediate-rates})};
\draw[->, thick] (poly) -- (6.5, -0.35);

\end{tikzpicture}
}
\caption{Sample complexity landscape illustrating the spectrum from linear to quadratic bounds based on the structural properties of the query family.}
\label{fig:complexity_landscape}
\end{figure}

\paragraph{Takeaways.} We believe our results yield conceptual insights that extend beyond our stylized setting and capture natural phenomena in learning from heterogeneous data sources.
\begin{itemize}[leftmargin=*, topsep=0pt]
\item \textbf{Learnability does not determine sample complexity.}
The co-occurrence graph gives a sharp qualitative characterization of
learnability, but it does not determine the quantitative rate. Among PAC
learnable query families, the optimal sample complexity can range from
$\widetilde{\Theta}(n/\epsilon^2)$ to
$\widetilde{\Theta}(n^2/\epsilon^2)$, with intermediate polynomial rates also
possible. Thus, once enough overlap exists to make learning possible, finer
structural properties of $\sellers$ still govern the cost of learning.

\item \textbf{Redundant overlap can dilute signal.}
The quadratic lower bound is witnessed by query families in which every
queryable set contains a common high-mass element $M$ together with only a
constant number of other elements. When the target distribution assigns
constant mass to $M$ and spreads the remaining mass roughly uniformly over the
rest of the domain, most conditional samples reveal only the common element.
The overlap is therefore abundant but uninformative: it repeatedly exposes the
same dominant content while obscuring the low-mass distinctions needed for
learning.

\item \textbf{Hierarchical structure enables targeted comparisons.}
Near-linear rates become possible when the query family supports recursive
localization. Hierarchical comparability formalizes this principle: the domain
can be partitioned in a tree-like way, and the query family contains local
queries that let the learner compare and refine mass estimates within
successively smaller regions. This structure turns overlap into useful
comparisons rather than redundant sampling.
\end{itemize}

\subsection{Related Work}\label{sec:relatedwork}
We now overview several areas that are related to our work.

\paragraph{Conditional sampling.}
The conditional sampling model was introduced in distribution testing by \cite{CanonneRonServedio2015} and independently by \cite{ChakrabortyFischerGoldhirshMatsliah2016}. In the full $\COND$ model, the algorithm may query arbitrary subsets and receives samples from the target conditioned on the queried set. This access can yield polylogarithmic or even dimension-independent bounds for several testing, evaluation, and mass-estimation tasks; see the survey of \cite{Canonne2020Survey} and recent work of \cite{AdarFischerLevi2025}.
In particular, \cite{ChakrabortyFischerGoldhirshMatsliah2016,AdarFischerLevi2025} study the related problem of learning a distribution up to permutations (for well-structured $\sellers$) and obtain $\poly(\log(n))$ bounds. Note that this does \emph{not} contradict our lower bound, since this is an easier task than learning in TV distance.
Several variants of restricted conditional sampling have been studied, including non-adaptive conditional sampling \citep{KamathTzamos2019}, interval or subcube, and pairwise comparison models. Our contribution is not another algorithm for a given restriction, but a structural theory for arbitrary fixed query families.

\paragraph{Data mixtures and data markets.}
Our motivation is related to recent work on data selection for large-scale
model training and on data markets. In language-model pretraining, several
works study how the mixture proportions of different data domains affect model
performance \citep{xie2023doremi,ye2025datamixing}. Work on
data markets and data acquisition studies how training data should be procured
or valued when it is supplied by multiple strategic or heterogeneous providers
\citep{lu2024dataacquisition,huang2022skill}. Our model
isolates a basic statistical question that underlies these settings: when data
sources expose different overlapping parts of an unknown population, how does
the overlap pattern determine whether the global distribution can be learned,
and at what sample cost?

\paragraph{Luce, Plackett--Luce, and comparison models.}
On a positive support, the conditional law
\(p(x\mid S)=p_x/\sum_{y\in S}p_y\) is the Luce choice rule on menu \(S\)
\cite{Luce1959}. Thus our model is a sampling-based version of stochastic choice
from an arbitrary collection of menus. \cite{AlosFerrerMihm2025}
characterize Luce rationalizability for arbitrary menu collections, including
settings with zero choice probabilities, and discuss implications for
falsification, identification, and prediction. Their setting is structural and
passive, whereas we study adaptive sampling and finite-sample learning
guarantees.

The pairwise special case coincides with the Bradley--Terry--Luce comparison model \citep{BradleyTerry1952,Luce1959}. A large literature studies passive or designed estimation from pairwise comparisons, including topology-dependent minimax rates \citep{ShahBalakrishnanBradleyParekhRamchandranWainwright2016}, Rank Centrality \citep{NOS17}, and spectral ranking methods \citep{MaystreGrossglauser2015,AgarwalPatilAgarwal2018}. These works typically estimate latent quality scores or rankings under a given
comparison graph, often under regularity assumptions such as positive weights,
bounded dynamic range, or fixed sampling designs. By contrast, our loss is total variation distance to the normalized distribution, and our main question is how the allowed conditional queries determine both learnability and labeled distribution-learning complexity.

\paragraph{Local sampling and active subset-wise feedback.}
\cite{FotakisKalavasisTzamos2022} study perfect sampling from pairwise comparisons in a local-sampling model where comparison sets are drawn from an exogenous design. Our generic complete-graph upper bound uses a related Markov-chain idea, but the query family is chosen adaptively by the learner and the goal is PAC learning rather than exact sampling. Active subset-wise Plackett--Luce feedback has also been studied for ranking and best-item objectives \citep{SahaGopalan2019}; those objectives differ from learning the entire labeled distribution in TV.

\paragraph{Biased sampling and stochastic choice.}
Our model is also related to classical biased-sampling and selection-bias models.
In that literature, one observes samples from known biased versions of an unknown
distribution, often with density proportional to \(w_i(x)\,dP(x)\) for known weight
functions \(w_i\). Querying a provider set \(S\) in our model corresponds to the
special case \(w_S(x)=\mathbf{1}\{x\in S\}\). \cite{Var85} introduced
nonparametric estimation for such selection-bias models, and
\cite{GVW88} developed large-sample theory for the
NPMLE in multi-sample biased-sampling models, including identifiability,
consistency, and asymptotic efficiency results. In stochastic-choice language,
the conditional law \(p(x\mid S)=p_x/\sum_{y\in S}p_y\) is the Luce or
multinomial-logit choice rule on menu \(S\)~\cite{Luce1959}; recent work also studies
Luce rationalizability for arbitrary collections of menus~\cite{AlosFerrerMihm2025}. These
works are closest to our qualitative identifiability results. Our focus is
different: we study adaptive conditional-query access to an arbitrary fixed
family, distribution-free PAC learning in total variation, and finite-sample
minimax rates.

\paragraph{Multi-distribution learning.}
Multi-distribution learning was introduced by \cite{BHPQ17} under the name collaborative PAC learning. This line of work studies supervised learning, with a focus on binary classification, from \(k\) labeled distributions, with the goal of achieving small classification error \emph{simultaneously} across them. \cite{HJZ22} established optimal agnostic rates for finite hypothesis classes, while \cite{ZZC+24} and \cite{Pen24} obtained nearly optimal rates for classes of finite VC dimension. More recently, \cite{HSJ26} studied multi-distribution learning under bounded label noise and showed that the fast rates available in the single-distribution setting need not extend to multiple distributions. In this literature, the different distributions are part of the learning objective: the learner is evaluated by its classification error on each distribution or by its worst-case error across them. Our setting has a different objective. There is a single target distribution \(p\), and the provider-specific conditional distributions \(p(\cdot\mid S)\) describe only the learner's restricted access to \(p\). The learner is not evaluated separately on these conditional distributions; rather, it must recover their relative normalizations and reconstruct the single global distribution \(p\) in total variation distance. Moreover, the focus of our setting is distribution learning whereas the focus of that line of work is mostly binary classification.

\section{Preliminaries}\label{sec:prelim}
Let $[n] = \{1, \dots, n\}$ denote the finite universe of elements (e.g., a vocabulary of tokens). 
The objective to learn is a discrete probability distribution $p$ over $[n]$, belonging to the probability simplex
\begin{equation*}
    \Delta([n]) := \left\{ p \in \mathbb{R}_{\ge 0}^n : \sum_{x=1}^n p_x = 1 \right\}.
\end{equation*}

We model the heterogeneous data providers as a fixed query family $\mathcal{S} = \{S_1, \dots, S_k\}$, where each $S_i \subseteq [n]$ is non-empty. Each $S_i$ represents the part of the domain the $i$-th data provider can generate samples from. For any set $S \subseteq [n]$, we denote its total probability mass under the target distribution as $p(S) := \sum_{x \in S} p_x$.

The learner may not be able to
sample directly from the global distribution $p$. Instead,
it can only query provider sets $S\in\sellers$. The family $\mathcal S$ is fixed in advance, but at each round the learner may
adaptively choose which set $S\in\mathcal S$ to query. The sets in
$\mathcal S$ are the atomic query options: after selecting $S$, the learner
receives one sample from $p(\cdot\mid S)$ and cannot impose an additional
condition $T\subsetneq S$, unless $T$ is itself an allowed set in
$\mathcal S$.
We use the UNIF0 convention for
zero-mass queries\footnote{To handle edge cases where a queried provider's set has zero mass under the target distribution ($p(S) = 0$), we adopt the standard UNIF0 convention used in the conditional sampling literature, where the oracle returns a uniformly random point from the queried set ($Q_{p,S} := u_S$). An alternative is the FAIL convention, where the oracle returns a failure symbol. The two conventions lead to the same structural criteria and sample bounds in all cases, so we proceed with UNIF0 for clarity.}. Thus, for every nonempty $S\subseteq[n]$, define the
one-query response law $Q_{p,S}$ on $[n]$ by
\[
    Q_{p,S}(x)
    :=
    \begin{cases}
        \dfrac{p_x}{p(S)}\,\mathbf 1\{x\in S\},
        & p(S)>0, \\[1.25em]
        \dfrac{1}{|S|}\,\mathbf 1\{x\in S\},
        & p(S)=0.
    \end{cases}
\]
Equivalently, if $p(S)>0$, then $Q_{p,S}=p(\cdot\mid S)$, while if
$p(S)=0$, then $Q_{p,S}=u_S$, the uniform distribution on $S$.
We denote the family of one-query response laws accessible to the learner as $\text{Obs}_{\mathcal{S}}(p) := (Q_{p,S})_{S \in \mathcal{S}}$.

Throughout this work, all probability spaces are finite.  For two distributions $P,Q$ on the same finite set $\Omega$, let
\[
  \TV(P,Q):=\frac12\sum_{\omega\in\Omega}|P(\omega)-Q(\omega)|
\]
be the total variation distance and
\[
  \KL(P\|Q):=\sum_{\omega\in\Omega}P(\omega)
  \log\frac{P(\omega)}{Q(\omega)}
\]
the Kullback--Leibler divergence, with the usual convention that the summand is $0$ when $P(\omega)=0$ and the divergence is $+\infty$ if there exists $\omega$ with $Q(\omega) = 0, P(\omega) > 0.$ All logarithms are natural.  If $V,X$ are finite random variables, $H(V)$ denotes entropy and $I(V;X)$ denotes mutual information.
If $V$ and $X$ are random variables taking values in finite sets $\mathcal{V}$ and $\mathcal{X}$ with joint distribution $P_{V,X}$ and marginals $P_V$ and $P_X$, the Shannon entropy of $V$ is
\[
  H(V) := \sum_{v \in \mathcal{V}} P_V(v) \log \frac{1}{P_V(v)},
\]
and their mutual information is defined as
\[
  I(V;X) := \KL(P_{V,X} \| P_V \otimes P_X),
\]
where $P_V \otimes P_X$ denotes the product distribution $(P_V \otimes P_X)(v,x) := P_V(v)\cdot P_X(x)$.
Moreover, for two strings $\sigma,\tau\in\mathcal A^m$ over a common alphabet $\mathcal A$,
their Hamming distance is
\[
  d_H(\sigma,\tau)
  :=
  \bigl|\{j\in[m]: \sigma_j\neq \tau_j\}\bigr|
  =
  \sum_{j=1}^m \mathbf 1\{\sigma_j\neq \tau_j\}.
\]

We will use the following elementary Bernoulli concentration bound.  If $B_1,\ldots,B_N$ are independent random variables in $\{0,1\}$ with common mean $\theta$, then for every $u>0$,
\[
  \Prob\left[\left|\frac1N\sum_{t=1}^N B_t-\theta\right|\ge u\right]
  \le 2\exp(-2Nu^2).
\]

\paragraph{Notions of learnability.}
Throughout our work we study two notions of learnability for a target class $\mathcal{P} \subseteq \Delta([n])$ and query family $\sellers$. Given $U \subseteq [n]$, we denote \[\mathcal P_U^+
    :=
    \left\{
        p\in\Delta([n]) : p_x>0 \text{ iff } x\in U
    \right\} \quad \text{and} \quad \mathcal P_U
    :=
    \left\{
        p\in\Delta([n]) : p_x=0 \text{ for all } x\notin U
    \right\},\] i.e., $\mathcal{P}_{U}^{+}$ consists of distributions with support exactly $U$, while $\mathcal{P}_{U}$ allows any support $ \subseteq U$. 
    
    We begin with the standard PAC framework, requiring a single sample bound that holds simultaneously for the entire class.

\begin{definition}[PAC learnability]\label{def:uniform}
Fix a target class $\mathcal{P} \subseteq \Delta([n])$. We say that $\mathcal{S}$ PAC learns $\mathcal{P}$ if there exists a sample-complexity function $m: (0,1)^2 \to \N$ such that for every $\epsilon, \delta \in (0,1)$ there is an adaptive randomized learner which, for every target distribution $p \in \mathcal{P}$ makes at most $m(\epsilon, \delta)$ queries and outputs $\hat{p}$ satisfying 
\begin{equation*}
    \mathbb{P}[d_{TV}(\hat{p}, p) \le \epsilon] \ge 1 - \delta \,,
\end{equation*}
where the probability is taken with respect to the randomness of the samples and the internal randomness of the learner.
The optimal PAC query complexity is denoted by
$q^*_{\sellers,\mathcal{P}}(\eps,\delta).$
\end{definition}

We also consider the weaker notion of pointwise consistency, which allows the required number of samples to achieve some given error bound to depend on the specific target distribution.

\begin{definition}[Pointwise consistency]\label{def:universal}
Fix a target class $\mathcal{P} \subseteq \Delta([n])$. We say that $\mathcal{S}$ pointwise consistently learns $\mathcal{P}$ if there exists a sequence of adaptive randomized learners $(A_t)_{t\ge 1}$, where $A_t$ makes $t$ queries on $\mathcal{S}$, such that for every $p \in \mathcal{P}$ and every $\epsilon, \delta \in (0,1)$ there exists $t_0 = t_0(p, \epsilon, \delta)$ such that for all $t \geq t_0$ 
\begin{equation*}
    \mathbb{P}_p[d_{TV}(A_t, p) \le \epsilon] \ge 1 - \delta \,. 
\end{equation*}
\end{definition}

Note that PAC learnability implies pointwise consistency.

\begin{proposition}\label{prop:pac_to_pointwise}
If $\sellers$ PAC learns a class $\mathcal{P}$, then $\sellers$ pointwise consistently learns $\mathcal{P}$.
\end{proposition}

\begin{proof}
For each $s\ge 1$, let $B_s$ be a PAC learner with parameters
$(\varepsilon_s,\delta_s)=(1/s,1/s)$, and let $N_s$ be its query bound. Define
$M_s:=\max_{1\le j\le s}N_j.$
Given budget $k$, let
$s(k):=\max\{s\le k:M_s\le k\},$
if this set is nonempty; otherwise output an arbitrary fixed distribution. If
$s(k)$ is defined, let $A_k$ run $B_{s(k)}$. Then $A_k$ uses at most
$M_{s(k)}\le k$ queries. Also $s(k)\to\infty$, since for every fixed $s$,
we have $s(k)\ge s$ whenever $k\ge \max\{s,M_s\}$.
Fix $p\in\mathcal P$ and $\varepsilon,\delta\in(0,1)$. Choose $s_0$ with
$1/s_0\le \min\{\varepsilon,\delta\}$. For all large enough $k$,
$s(k)\ge s_0$, and hence
\[
  \Pr_p[\TV(A_k,p)\le \varepsilon]
  \ge
  \Pr_p[\TV(B_{s(k)},p)\le 1/s(k)]
  \ge
  1-\frac1{s(k)}
  \ge
  1-\delta. \qedhere
\]
\end{proof}

We defer several technical lemmas that are useful for our results to 
\cref{app:tech-lemmas}.

\section{Learnability and Query Families}\label{sec:learnability}

We first characterize the qualitative limits of learning from a restricted collection of providers. In other words, we answer the following basic question: given a fixed query family $\sellers$, which distribution classes $\mathcal{P}$ can be learned from conditional queries to $\sellers$?

The answer is governed by how the available sets $S \in \sellers$ overlap on the target support. Informally, if two elements $x, y \in [n]$ appear together in some queryable set $S$, conditional samples from $S$ directly reveal their relative probability mass. If the target domain is connected via a chain of such overlaps, these local comparisons can be propagated to calibrate the entire distribution. We formalize this overlap structure through the \emph{co-occurrence graph}.

As we will show, the topology of this graph characterizes learnability. For pointwise consistency on a known support, mere \emph{connectivity} of the graph is necessary and sufficient. However, for PAC learning, where a single finite-sample bound must hold even for distributions that place arbitrarily small mass on intermediate bridge points, connectivity fails, and \emph{completeness} of the graph is required.

\begin{definition}[Co-occurrence graph]\label{def:cograph}
    For a given subset $U \subseteq [n]$ and query family $\sellers$, the co-occurrence graph $\co(\mathcal{S}, U)$ is defined as the graph with vertex set $U$ and edge set
    \begin{equation*}
        E\left(\co(\sellers, U)\right) := \big\{ \{x, y\} \subseteq U : \exists S \in \sellers \text{ with } x, y \in S \big\} \text{.}
    \end{equation*}
\end{definition}
In words, an edge $\{x,y\}$ indicates that some provider set contains both $x$ and $y$, and
therefore allows the learner to compare their relative masses through 
conditional queries to $S$. 
We say that $\co(\sellers, U)$ is \emph{connected} if there exists a path between every pair of elements of $U$. Furthermore, $\co(\sellers, U)$ is \emph{complete} if there exists an edge between every pair of distinct elements of $U$, i.e., every possible pair co-occurs within the domain of at least one provider.

\subsection{Pointwise Consistency: Connected Support}

We first consider the ability of a query family to pointwise consistently learn a target distribution. In this regime, the learner knows the exact support of the target distribution, and the sample complexity is allowed to depend on the specific target. Our main result here shows that 
connectedness of the co-occurrence graph is necessary and sufficient. Before providing the main proof, we give some intuition for the result. 

If $\co(\sellers,U)$ is connected, the learner fixes a spanning tree of this
graph. Each tree edge $\{u,v\}$ lies in some set $S\in\sellers$, and
conditional samples from $S$ can estimate the ratio $p_v/p_u$. By
multiplying these estimated ratios along paths from a root $r$, the learner
estimates $p_x/p_r$ for every $x\in U$. Conversely, suppose $\co(\sellers,U)$ is disconnected, with connected
components $C_1,\dots,C_m$, $m\ge2$. Then, no provider set $S$ can ``compare'' points $x,y$ from
different components. Hence, the learner can estimate masses \emph{within} each
component, but it cannot estimate the relative total weights of the components. 

\begin{theorem}
\label{thm:market_pointwise}
Fix a nonempty set $U\subseteq[n]$. Then,
 $\sellers$ pointwise consistently learns
$\mathcal P_U^+$ if and only if the co-occurrence graph
$\co(\sellers,U)$ is connected. 
\end{theorem}

\begin{proof}
We prove the two directions separately.

\paragraph{Sufficiency.}
Assume first that $\co(\sellers,U)$ is connected. If $|U|=1$, the learner
outputs the point mass on the unique element of $U$, so assume $|U|\ge 2$.

Choose a spanning tree $\mathcal T$ of $\co(\sellers,U)$ and root it at some
vertex $r\in U$. Orient every tree edge away from $r$. For each directed tree
edge $e=(u\to v)$, choose one provider set
$
    S_e\in\sellers
$ and $
    u,v\in S_e .
$
Such a set exists by the definition of the co-occurrence graph.
Let
\[
    E_{\mathcal T}:=|\mathrm{Edges}(\mathcal T)|=|U|-1.
\]
We now define a sequence of learners $(A_t)_{t\ge 1}$ satisfying the query
counting convention in Definition~\ref{def:universal}: learner $A_t$ makes
exactly $t$ queries.

Given the total budget $t$, let
\[
    k_t:=\left\lfloor \frac{t}{E_{\mathcal T}}\right\rfloor .
\]
The learner queries each tree-edge provider set $S_e$ exactly $k_t$ times.
This uses $E_{\mathcal T}k_t\le t$
queries. 

For a directed tree edge $e=(u\to v)$ and endpoint $y\in\{u,v\}$, let $N_e^{(t)}(y)$
be the number of times the useful $k_t$ samples from $S_e$ return $y$, and define
\[
    \widehat\theta_e^{(t)}(y)
    :=
    \frac{N_e^{(t)}(y)}{k_t},
    \qquad y\in\{u,v\},
\]
whenever $k_t\ge 1$. If $k_t=0$, define the estimate arbitrarily; this affects
only finitely many values of $t$ and is irrelevant for consistency. For
$k_t\ge1$, define the lower-truncated frequencies
\[
    \widetilde\theta_e^{(t)}(y)
    :=
    \max\left\{\widehat\theta_e^{(t)}(y),\frac1{k_t}\right\},
    \qquad y\in\{u,v\}.
\]
and also let $\widehat\rho_r^{(t)}:=1.$
For $x\in U\setminus\{r\}$, let
\[
    r=v_0,v_1,\dots,v_\ell=x
\]
be the unique path from $r$ to $x$ in $\mathcal T$. Writing
$e_i=(v_{i-1}\to v_i)$ for the directed edges on this path, define
\[
    \widehat\rho_x^{(t)}
    :=
    \prod_{i=1}^{\ell}
    \frac{\widetilde\theta_{e_i}^{(t)}(v_i)}
         {\widetilde\theta_{e_i}^{(t)}(v_{i-1})}.
\]
Finally, define the normalized values
\[
    \widehat Z_t:=\sum_{z\in U}\widehat\rho_z^{(t)},
    \qquad
    \widehat p_x^{(t)}
    :=
    \frac{\widehat\rho_x^{(t)}}{\widehat Z_t}
    \quad (x\in U),
\]
and set $\widehat p_x^{(t)}:=0$ for $x\notin U$. The learner $A_t$ outputs
$\widehat p^{(t)}$.

Fix any target distribution $p\in\mathcal P_U^+$. Since $p$ has exact support
$U$, every point in $U$ has positive mass. In particular, $p_r>0.$
Moreover, for every tree edge $e=(u\to v)$, the chosen set $S_e$ contains both
$u$ and $v$, so $p(S_e)\ge p_u+p_v>0.$
Thus, every useful query to $S_e$ returns an ordinary conditional sample from
$Q_{p,S_e}=p(\cdot\mid S_e)$.

For $y\in\{u,v\}$, write
\[
    \theta_e(y)
    :=
    p(y\mid S_e)
    =
    \frac{p_y}{p(S_e)}.
\]
Because $p_y>0$ and $p(S_e)>0$, we have $\theta_e(y)>0$.

Couple the learners across stages by imagining, for each tree edge $e$, an
infinite i.i.d. stream of samples from $Q_{p,S_e}$, with learner $A_t$ using
the first $k_t$ samples from each stream. Since $k_t\to\infty$, we have that
\[
    \widehat\theta_e^{(t)}(y)
    \longrightarrow
    \theta_e(y)
    \qquad\text{almost surely}
\]
for every tree edge $e$ and every endpoint $y$ of $e$.

The lower truncation by $1/k_t$ does not change the limit. Indeed,
$1/k_t\to0$, while $\theta_e(y)>0$. Hence, almost surely, for all sufficiently
large $t$,
\[
    \widehat\theta_e^{(t)}(y)>\frac{\theta_e(y)}2
    \qquad\text{and}\qquad
    \frac1{k_t}<\frac{\theta_e(y)}2.
\]
Therefore, eventually, $\widetilde\theta_e^{(t)}(y) = \widehat\theta_e^{(t)}(y),$
and so
$\widetilde\theta_e^{(t)}(y)
    \longrightarrow
    \theta_e(y)
    \text{ almost surely}.$

Consequently, for every directed tree edge $e=(u\to v)$,
\[
    \frac{\widetilde\theta_e^{(t)}(v)}
         {\widetilde\theta_e^{(t)}(u)}
    \longrightarrow
    \frac{\theta_e(v)}{\theta_e(u)}
    =
    \frac{p(v\mid S_e)}{p(u\mid S_e)}
    =
    \frac{p_v}{p_u}
    \qquad\text{almost surely}.
\]
Multiplying these convergent edge ratios along the unique path from $r$ to
$x$ gives
\[
    \widehat\rho_x^{(t)}
    \longrightarrow
    \prod_{i=1}^{\ell}\frac{p_{v_i}}{p_{v_{i-1}}}
    =
    \frac{p_x}{p_r}
    \qquad\text{almost surely}.
\]
For the root, this same result holds because
\[
    \widehat\rho_r^{(t)}=1=\frac{p_r}{p_r}.
\]
Thus, if we define the true ratio vector
\[
    \rho_x:=\frac{p_x}{p_r},
    \qquad x\in U,
\]
then
$
    \widehat\rho_x^{(t)}\to \rho_x
    \text{ almost surely for every }x\in U.
$
Since $U$ is finite,
\[
    \widehat Z_t
    =
    \sum_{z\in U}\widehat\rho_z^{(t)}
    \longrightarrow
    \sum_{z\in U}\rho_z
    =
    \sum_{z\in U}\frac{p_z}{p_r}
    =
    \frac1{p_r}
    \qquad\text{almost surely}.
\]
This is where the root mass is recovered: even though we set
$\widehat\rho_r^{(t)}=1$, the normalized estimate satisfies
\[
    \widehat p_r^{(t)}
    =
    \frac{1}{\widehat Z_t}
    \longrightarrow
    p_r
    \qquad\text{almost surely}.
\]
For every $x\in U$, we therefore obtain
\[
    \widehat p_x^{(t)}
    =
    \frac{\widehat\rho_x^{(t)}}{\widehat Z_t}
    \longrightarrow
    \frac{p_x/p_r}{1/p_r}
    =
    p_x
    \qquad\text{almost surely}.
\]
For $x\notin U$, both $\widehat p_x^{(t)}$ and $p_x$ are identically zero.
Hence
\[
    d_{TV}\bigl(\widehat p^{(t)},p\bigr)
    \longrightarrow 0
    \qquad\text{almost surely}.
\]

\paragraph{Necessity.}
If $\co(\sellers, U)$ is disconnected, write its connected components as $C_1,\ldots,C_r$ with $r\ge2$. Choose any $p\in\mathcal P_U^+$ and positive scalars $a_1,\ldots,a_r$ not all equal. Define
\[
 q_x=\frac{a_jp_x}{\sum_{\ell=1}^r a_\ell p(C_\ell)}
 \qquad\text{for }x\in C_j.
\]
Then $q\in\mathcal P_U^+$ and $q\ne p$. Every allowed query $S$ intersects $U$ in at most one connected component; otherwise it would create an edge between two components. Therefore conditioning on $S$ cancels the common multiplicative factor $a_j$, and $q(\cdot\mid S)=p(\cdot\mid S)$ for every informative query. Hence the map is not injective.
\end{proof}

\begin{remark}[Unknown support requires completeness] \label{rem:unknown_support}
If the exact support is unknown (the target class is $\mathcal{P}_U$), pointwise consistency requires $\co(\sellers,U)$ to be complete. If two elements never co-occur directly, connecting them requires intermediate bridge points. If the distribution assigns zero mass to these bridges, the path is  ``broken,'' making their relative weights unobservable. Indeed, suppose $\co(\sellers,U)$ is not complete. Then there exist distinct elements $a,b\in U$ that never co-occur in any queryable set $S\in\sellers$. The unknown-support class $\mathcal{P}_U$ contains the two-point family of distributions parameterized by $\theta$: $p^\theta = \theta\delta_a + (1-\theta)\delta_b,  \theta\in(0,1).$
Because $a$ and $b$ never appear together in any allowed query set, every query $S \in \sellers$ intersects $\{a, b\}$ in at most one element. 

Consequently, if the learner queries a set $S$ containing $a$ (which implies $b \notin S$), the conditional distribution $p^\theta(\cdot \mid S)$ assigns probability $1$ to $a$, regardless of $\theta$. Symmetrically, if $S$ contains $b$, the conditional distribution assigns probability $1$ to $b$, again regardless of $\theta$. If $S$ contains neither $a$ nor $b$, the probability of the set is zero. In all valid cases, the response law of every allowed query in $\sellers$ is  independent of $\theta$. Because the samples provide no information about $\theta$, the relative weight between $a$ and $b$ is unobservable, making it impossible to pointwise consistently learn the distribution. 
\end{remark}

\subsection{PAC Learnability: Complete Support}

For PAC learning, the learner must achieve a guaranteed error bound within a fixed sample size, simultaneously for all distributions in the class. As foreshadowed in \cref{rem:unknown_support}, mere connectivity is insufficient to provide such a uniform guarantee. Even if the distribution has full support on $U$, it might assign arbitrarily small positive mass to intermediate bridge points. At any fixed sample size, these low-mass bridges become statistically invisible, breaking the propagation of local ratios. Thus, PAC learnability requires the co-occurrence graph to be complete.

\begin{theorem}
\label{thm:market_pac}
Fix a nonempty set $U\subseteq[n]$. Then $\sellers$ PAC learns $\mathcal P_U$ if and only if the co-occurrence graph $\co(\sellers,U)$ is complete.
\end{theorem}

\begin{proof}We prove the two directions separately.

\paragraph{Sufficiency.} If $\co(\sellers, U)$ is complete, PAC learnability follows from the learner in \Cref{alg:mc_sampling}, proved in \Cref{thm:generic-upper-main}.

\paragraph{Necessity.} Choose distinct $a,b\in U$ that never co-occur in an allowed query set.
First suppose $U=\{a,b\}$. Consider
\[
  p=\frac23\delta_a+\frac13\delta_b,
  \qquad
  q=\frac13\delta_a+\frac23\delta_b .
\]
Every allowed query contains at most one of $a,b$. If a query contains exactly
one of them, then under both targets the conditional response is deterministic
on that point. If it contains neither, then the queried set has zero mass under
both targets, so the response law is target-independent under either zero-mass
convention. Hence every adaptive strategy has exactly the same transcript law
under $p$ and $q$. Since
\[
  \TV(p,q)=\frac13,
\]
PAC learning is impossible.

Now assume $m:=|U|\ge 3$. Let
\[
  R:=U\setminus\{a,b\}.
\]
Choose $\eta>0$ small enough that
\[
  \lambda:=1-(m-2)\eta\ge \frac12.
\]
Define $p^\eta,q^\eta\in\mathcal P_U^+.$ by
\[
  p^\eta_a=\frac{2\lambda}{3},
  \qquad
  p^\eta_b=\frac{\lambda}{3},
  \qquad
  p^\eta_z=\eta \quad (z\in R),
\]
and
\[
  q^\eta_a=\frac{\lambda}{3},
  \qquad
  q^\eta_b=\frac{2\lambda}{3},
  \qquad
  q^\eta_z=\eta \quad (z\in R).
\]
These are valid distributions with exact support $U$, because
\[
  \frac{2\lambda}{3}+\frac{\lambda}{3}+(m-2)\eta
  =
  \lambda+(m-2)\eta
  =
  1,
\]
and all masses are positive. Moreover,
\[
  \TV(p^\eta,q^\eta)
  =
  \frac12\left(
    \left|\frac{2\lambda}{3}-\frac{\lambda}{3}\right|
    +
    \left|\frac{\lambda}{3}-\frac{2\lambda}{3}\right|
  \right)
  =
  \frac{\lambda}{3}
  \ge
  \frac16.
\]

We next show that every one-query response law is close under these two
targets. Fix an allowed query set $S\in\cS$.
If $S$ contains neither $a$ nor $b$, then the restrictions of $p^\eta$ and
$q^\eta$ to $S$ are identical. Thus, if $S$ has positive mass, the conditional
response laws are identical; if $S$ has zero mass, the response law is again
target-independent under either zero-mass convention.

By the choice of $a,b$, no allowed query contains both $a$ and $b$. It remains
to consider the case where $S$ contains exactly one of them. Suppose first that
$a\in S$ and $b\notin S$; the case $b\in S$ and $a\notin S$ is symmetric. Let
\[
  c_S:=\sum_{z\in S\cap R}\eta
      =\eta |S\cap R|
      \le (m-2)\eta .
\]
Define
\[
  A:=\frac{2\lambda}{3},
  \qquad
  B:=\frac{\lambda}{3}.
\]
Then,
\[
  p^\eta(S)=A+c_S,
  \qquad
  q^\eta(S)=B+c_S.
\]
The conditional response laws on $S$ are therefore the normalized laws
\[
  Q^\star_{p^\eta,S}
  =
  \frac{A}{A+c_S}\delta_a
  +
  \sum_{z\in S\cap R}\frac{\eta}{A+c_S}\delta_z,
\qquad
\text{and}
\qquad
  Q^\star_{q^\eta,S}
  =
  \frac{B}{B+c_S}\delta_a
  +
  \sum_{z\in S\cap R}\frac{\eta}{B+c_S}\delta_z.
\]
All points outside $U$ have zero probability under both laws. 

We now compute their total variation distance. Since $A>B$, we have
\[
  \frac{A}{A+c_S}-\frac{B}{B+c_S}
  =
  \frac{c_S(A-B)}{(A+c_S)(B+c_S)}.
\]
For each $z\in S\cap R$,
\[
  \frac{\eta}{B+c_S}-\frac{\eta}{A+c_S}
  =
  \frac{\eta(A-B)}{(A+c_S)(B+c_S)}.
\]
Summing over $z\in S\cap R$ gives
\[
  \sum_{z\in S\cap R}
  \left(
    \frac{\eta}{B+c_S}-\frac{\eta}{A+c_S}
  \right)
  =
  \frac{c_S(A-B)}{(A+c_S)(B+c_S)}.
\]
Therefore
\[
\begin{aligned}
  \TV(Q^\star_{p^\eta,S},Q^\star_{q^\eta,S})
  &=
  \frac12
  \left[
    \frac{c_S(A-B)}{(A+c_S)(B+c_S)}
    +
    \frac{c_S(A-B)}{(A+c_S)(B+c_S)}
  \right]  \\
  &=
  \frac{c_S(A-B)}{(A+c_S)(B+c_S)}.
\end{aligned}
\]
Substituting $A=2\lambda/3$ and $B=\lambda/3$ yields
\[
  \TV(Q^\star_{p^\eta,S},Q^\star_{q^\eta,S})
  =
  \frac{c_S(\lambda/3)}
       {(2\lambda/3+c_S)(\lambda/3+c_S)}.
\]
Since $c_S\ge 0$ and $\lambda\ge 1/2$,
\[
  (2\lambda/3+c_S)(\lambda/3+c_S)
  \ge
  \frac{2\lambda}{3}\cdot\frac{\lambda}{3}
  =
  \frac{2\lambda^2}{9}.
\]
Thus
\[
  \TV(Q^\star_{p^\eta,S},Q^\star_{q^\eta,S})
  \le
  \frac{c_S(\lambda/3)}{2\lambda^2/9}
  =
  \frac{3c_S}{2\lambda}
  \le
  3c_S
  \le
  3(m-2)\eta.
\]
 Hence
\[
  \sup_{S\in\cS}
  \TV(Q^\star_{p^\eta,S},Q^\star_{q^\eta,S})
  \le
  3(m-2)\eta .
\]

Suppose, for contradiction, that a uniform PAC learner exists for
$\mathcal P_U^+.$ with accuracy $\varepsilon_0:=1/24$
and confidence $2/3$. Let $T$ be its worst-case query bound at parameters
$(\varepsilon_0,1/3)$. Choose $\eta>0$ small enough that, in addition to
$\lambda\ge 1/2$,
\[
  3T(m-2)\eta<\frac16.
\]
By \Cref{lem:hybrid}, the transcript laws under targets $p^\eta$
and $q^\eta$ have total variation distance less than $1/6$.

On the other hand, the PAC guarantee gives
\[
  \Pr_{p^\eta}\!\left[
    \TV(\widehat p,p^\eta)\le \varepsilon_0
  \right]
  \ge \frac23
\qquad \text{and} \qquad
  \Pr_{q^\eta}\!\left[
    \TV(\widehat p,q^\eta)\le \varepsilon_0
  \right]
  \ge \frac23.
\]
Because
\[
  \TV(p^\eta,q^\eta)\ge \frac16
  >
  2\varepsilon_0
  =
  \frac1{12},
\]
the two balls
\[
  \{\mu:\TV(\mu,p^\eta)\le \varepsilon_0\},
  \qquad
  \{\mu:\TV(\mu,q^\eta)\le \varepsilon_0\}
\]
are disjoint. Therefore, for the event
$
  E:=\{\TV(\widehat p,p^\eta)\le \varepsilon_0\},
$
we have
\[
  \Pr_{p^\eta}(E)\ge \frac23,
  \qquad
  \Pr_{q^\eta}(E)\le \frac13.
\]
Thus, the two transcript laws must differ by at least $1/3$ in total variation,
contradicting the upper bound $<1/6$ we established before.
Hence, no uniform PAC learner exists when $\co(\sellers,U)$ is not complete.
\end{proof}

\section{Query Family Structure and Sample Complexity}\label{sec:sample_complexity}

The previous section established that for uniform PAC learning on an unknown support, qualitative learnability is characterized by the co-occurrence graph: learning is possible if and only if $\co(\sellers,U)$ is complete. We now ask whether this qualitative topological condition also governs the quantitative sample complexity. 

The answer turns out to be negative. Completeness captures \emph{whether} pairs can be calibrated somewhere, but it does not characterize \emph{how} ``cleanly'' those comparisons are exposed. If co-occurrences only happen inside broad, hub-dominated sets, the relevant signal is heavily diluted, leading to worst-case sample complexities. 
Conversely, when the query family supports clean, localized comparisons, the learner can recursively decompose the domain to isolate subcategories and achieve near-ideal efficiency; we formalize this property as hierarchical comparability.

We explore this surprisingly rich landscape in the remainder of the section. We first establish the ideal benchmark: maximal query power achieves the ordinary empirical rate of $\Theta(n/\eps^2)$ (\Cref{sec:full-cond}). Moving to restricted families, we show that completeness alone guarantees a $\widetilde O(n^2/\eps^2)$ sample complexity via a simulated Markov chain (\Cref{sec:generic-quadratic-upper-bound}). We prove this quadratic bound is worst-case optimal by constructing a hub-dominated family that requires $\Omega(n^2/\eps^2)$ queries (\Cref{thm:quadratic-lower-main}). However, we show that query families satisfying hierarchical comparability bypass this worst case to recover the near-linear $\widetilde O(n/\eps^2)$ rate (\Cref{sec:hierarchical}). Next, we bridge these extremes, proving the existence of complete query families whose optimal sample complexities tightly span the continuum $\widetilde{\Theta}(n^{1+\gamma}/\eps^2)$ for every $\gamma\in(0,1)$ (\Cref{sec:continuum}). Finally, we study the convergence rates of pointwise-consistent algorithms (\Cref{sec:pointwise-rates}).

\subsection{The Ideal Benchmark: Full Conditional Access}\label{sec:full-cond}
We first show the best possible rate when the learner has maximum query power. Folklore results show that by only querying $[n]$, the learner achieves TV error $\varepsilon$ using $O(n/\varepsilon^2)$ samples (see, e.g., \citep{canonne2020short}). Thus, it is natural to ask whether the additional conditional querying power gives it an advantage.
We show that this is not the case. 
\begin{theorem}
\label{thm:full-cond-main}
Let $\sellers=2^{[n]}\setminus\{\emptyset\}$. Then, for every sufficiently small $\epsilon$
it holds that
\[
 q^*_{\sellers, \mathcal P^+_{[n]}}(\eps,1/3)=\Omega\!\left(\frac{n}{\eps^2}\right).
\]
\end{theorem}

\begin{proof}
Without loss of generality, we assume that $n$ is even. Let $m=n/2$ and partition $[n]$ into pairs $P_j=\{2j-1,2j\}$. For $\sigma\in\bits^m$ and $\alpha\in(0,1/2]$, define
\[
 p^\sigma_{2j-1}=\frac{1+\alpha\sigma_j}{n},
 \qquad
 p^\sigma_{2j}=\frac{1-\alpha\sigma_j}{n}.
\]
Let $\cC\subseteq\bits^m$ be the code from \Cref{lem:code}. If $\sigma,\tau\in\cC$ are distinct, then
\[
 \TV(p^\sigma,p^\tau)
 =
 \frac{2\alpha}{n}d_H(\sigma,\tau)
 \ge \frac{\alpha}{4}.
\]
Set $\alpha=16\eps$, so for $\eps\le1/32$ all code distributions are full support and pairwise $4\eps$-separated.

Fix any nonempty query set $S\subseteq[n]$. For $i\in[n]$, define $\tau_i(\sigma)=\sigma_j$ if $i=2j-1$ and $\tau_i(\sigma)=-\sigma_j$ if $i=2j$. Let $P_{\sigma,S}$ be the conditional response distribution under $p^\sigma$ and let $U_S$ be uniform on $S$. For $i\in S$,
\[
 \frac{P_{\sigma,S}(i)}{U_S(i)}
 =
 \frac{1+\alpha \tau_i(\sigma)}
 {1+\alpha \overline\tau_S(\sigma)},
 \qquad
 \overline\tau_S(\sigma)=\frac1{|S|}\sum_{i\in S}\tau_i(\sigma).
\]
Since $|\tau_i|\le1$, $|\overline\tau_S|\le1$, and $\alpha\le1/2$,
\[
 \left|\frac{P_{\sigma,S}(i)}{U_S(i)}-1\right|\le 4\alpha.
\]
Using $x\log x\le (x-1)+(x-1)^2$ and the fact that the first-order term integrates to zero,
\[
 \KL(P_{\sigma,S}\|U_S)\le 16\alpha^2.
\]

Now let $\Sigma$ be uniform on $\cC$ and let $\Pi$ be the full transcript of any adaptive $q$-query learner, including its internal randomness. Conditional on any past history, the next query set is fixed. By \Cref{lem:mi-ref}, the mutual information between $\Sigma$ and the next response, conditional on that history, is at most $16\alpha^2$. The chain rule gives
\[
 I(\Sigma;\Pi)\le 16q\alpha^2.
\]

If the learner outputs $\widehat p$ with $\TV(\widehat p,p^\Sigma)\le\eps$ with probability at least $2/3$ for every codeword, then nearest-neighbor decoding recovers $\Sigma$ with probability at least $2/3$ because the code distributions are $4\eps$-separated. By \Cref{lem:fano},
\[
 I(\Sigma;\Pi)\ge c_1 m
\]
for a universal constant $c_1>0$ and all sufficiently large $m$. Therefore
\[
 q\ge c_2\frac{m}{\alpha^2}
 =
 \Omega\!\left(\frac{n}{\eps^2}\right). \qedhere
\]
\end{proof}

\subsection{Worst-Case Complete Query Families}\label{sec:generic-quadratic-upper-bound}
We now ask what quantitative rate can be guaranteed from the qualitative PAC condition alone. Suppose the query family $\sellers$ induces a complete co-occurrence graph on the target support $U$. Does this topological condition automatically yield the ordinary-sampling rate $\widetilde O(|U|/\eps^2)$?

We answer this negatively. While completeness is sufficient for PAC learnability, in the worst case it guarantees only a quadratic rate. The upper bound below establishes a generic $\widetilde O(|U|^2/\eps^2)$ learner for any complete query family; the subsequent lower bound (\Cref{thm:quadratic-lower-main}) proves this dependence is unavoidable without further structural assumptions.

\paragraph{Sampling via a simulated Markov chain.}
The algorithm achieves the quadratic bound by using conditional queries to simulate a reversible Markov chain whose stationary distribution is the unknown target $p$. Because $\co(\sellers,U)$ is complete, every pair $x,y\in U$ has at least one witness set $W_{xy}\in\sellers$ containing both points. Crucially, we symmetrize this choice so that $W_{xy} = W_{yx}$.

From a current state $x$, the random walk proposes a new state $y$ uniformly from $U\setminus\{x\}$, queries the witness set $W_{xy}$, and transitions to $y$ only if the oracle returns $y$. This process requires just one query per step and visits states in proportion to their true masses without the learner needing to estimate all conditional probabilities. The chain mixes in $O(|U|\log(1/\eps))$ steps. By independently restarting the walk and recording the terminal states, the learner obtains approximate independent samples from $p$; the empirical distribution of these samples gives the claimed PAC bound. The formal algorithm is described in \Cref{alg:mc_sampling}.

\begin{algorithm}[ht]
\SetAlgoLined
\DontPrintSemicolon
\KwIn{Universe $U$ of size $m\ge 2$, query access to witness sets $W_{xy}\in\sellers$, target accuracy $\eps>0$, confidence $\delta>0$.
Burn-in steps $B = \lceil (m-1)\log(2/\eps) \rceil$, sample count $M = \Theta((m+\log(1/\delta))/\eps^2)$.}
Initialize an empty multiset $S \gets \emptyset$\;
Fix an arbitrary starting state $x_0\in U$.\;
\For{$i = 1$ \KwTo $M$}{
Initialize $x\leftarrow x_0$.\;
    \For{$t = 1$ \KwTo $B$}{
        Draw a candidate state $y$ uniformly at random from $U \setminus \{x\}$\;
        Query $W_{xy}$ and receive a sample $z\sim Q_{p,W_{xy}}$\;
        \eIf{$z = y$}{
            $x \gets y$ \tcc*{Accept the proposed move}
        }{
            $x \gets x$ \tcc*{Reject and remain at current state}
        }
    }
    $S \gets S \cup \{x\}$ \tcc*{Record the terminal state}
}
\Return empirical distribution $\widehat{p}$ of $S$\;
\caption{Markov Chain Sampling for Target Distribution $p$}
\label{alg:mc_sampling}
\end{algorithm}

\begin{theorem}[Complete co-occurrence graph quadratic learner]
\label{thm:generic-upper-main}
Let $U\subseteq[n]$ be nonempty with $|U| = m$ and suppose $\co(\sellers,U)$ is complete. Then, $\sellers$ PAC learns $\mathcal P_U$ with query complexity
\[
 q^*_{\sellers,\mathcal P_U}(\eps,\delta)
 \le
 C\,\frac{m(m+\log(1/\delta))\log(2/\eps)}{\eps^2}
\]
for an absolute constant $C$. 
\end{theorem}

\begin{proof}
Assume that $\co(\sellers,U)$ is complete. We show that $\sellers$ PAC learns
\[
    \mathcal P_U
    =
    \{p\in\Delta([n]) : p_x=0 \text{ for all } x\notin U\}.
\]

Let $m:=|U|$. If $m=1$, the learner outputs the point mass on the unique
element of $U$, so assume $m\ge 2$. Since $\co(\sellers,U)$ is complete, for
every unordered pair $\{x,y\}\subseteq U$, fix a symmetric witness set
\[
    W_{xy}=W_{yx}\in\sellers
    \qquad\text{such that}\qquad
    x,y\in W_{xy}.
\]

Fix an arbitrary target $p\in\mathcal P_U$. The learner simulates a Markov chain
on the state space $U$. From current state $x\in U$, it chooses a challenger
$y\in U\setminus\{x\}$ uniformly at random, queries $W_{xy}$, and moves to
$y$ if and only if the oracle response is $y$. Otherwise it stays at $x$.

Equivalently, for $x\neq y$,
\[
    K(x,y)
    =
    \frac{1}{m-1}Q_{p,W_{xy}}(y),
\]
where $Q_{p,W_{xy}}$ is the one-query response law from the preliminaries. The
diagonal probability is
\[
    K(x,x)
    =
    1-\sum_{y\in U\setminus\{x\}}K(x,y).
\]
This definition is valid for every $p\in\mathcal P_U$, including targets whose
support is a strict subset of $U$, because $Q_{p,S}$ is defined even when
$p(S)=0$.

We first show that $p$ is stationary for $K$. It suffices to verify detailed
balance. Fix distinct $x,y\in U$. If $p(W_{xy})=0$, then since
$x,y\in W_{xy}$ and $p$ is nonnegative, we have
\[
    p_x=p_y=0.
\]
Therefore
\[
    p_xK(x,y)=p_yK(y,x)=0.
\]
If $p(W_{xy})>0$, then the zero-mass convention is irrelevant and
\[
    Q_{p,W_{xy}}(y)=\frac{p_y}{p(W_{xy})},
    \qquad
    Q_{p,W_{xy}}(x)=\frac{p_x}{p(W_{xy})}.
\]
Thus
\[
    p_xK(x,y)
    =
    \frac{1}{m-1}\frac{p_xp_y}{p(W_{xy})}
    =
    p_yK(y,x).
\]
Hence detailed balance holds for every pair $x\neq y$, and so $p$ is
stationary for $K$.

Next we prove a uniform minorization. We claim that for every $x\in U$,
\[
    K(x,\cdot)
    \ge
    \frac{1}{m-1}p(\cdot)
\]
pointwise on $U$.

First fix $y\neq x$. If $p(W_{xy})>0$, then
\[
    K(x,y)
    =
    \frac{1}{m-1}\frac{p_y}{p(W_{xy})}
    \ge
    \frac{p_y}{m-1},
\]
because $p(W_{xy})\le 1$. If $p(W_{xy})=0$, then $p_y=0$, so the same
inequality holds trivially.

It remains to check the diagonal coordinate. If $p_x=0$, then
\[
    K(x,x)\ge 0=\frac{p_x}{m-1}.
\]
If $p_x>0$, then for every challenger $y\neq x$, the set $W_{xy}$ has positive
mass because it contains $x$. Therefore
\[
    1-Q_{p,W_{xy}}(y)
    \ge
    Q_{p,W_{xy}}(x)
    =
    \frac{p_x}{p(W_{xy})}
    \ge
    p_x.
\]
Averaging over the uniformly chosen challenger gives
\[
    K(x,x)
    =
    \frac{1}{m-1}
    \sum_{y\neq x}
    \bigl(1-Q_{p,W_{xy}}(y)\bigr)
    \ge
    p_x
    \ge
    \frac{p_x}{m-1}.
\]
This proves the minorization.

The minorization implies rapid mixing (\cref{lem:doeblin-minorization}). If $m=2$, then
\[
    K(x,\cdot)\ge p(\cdot)
\]
for every $x\in U$, and since both sides have total mass one, this implies
$K(x,\cdot)=p(\cdot)$ for every $x$. Thus one step already has distribution
$p$.

For $m>2$, define $\Pi_p$ to be the rank-one kernel whose every row is $p$.
The minorization gives the Doeblin decomposition
\[
    K
    =
    \frac{1}{m-1}\Pi_p
    +
    \left(1-\frac{1}{m-1}\right)R,
\]
where $R$ is a Markov kernel. Consequently, for every starting distribution
$\mu$ on $U$,
\[
    d_{TV}(\mu K^t,p)
    \le
    \left(1-\frac{1}{m-1}\right)^t
    \le
    e^{-t/(m-1)}.
\]

Set
\[
    B
    :=
    \left\lceil
        (m-1)\log\frac{2}{\eps}
    \right\rceil .
\]
Starting from any fixed state in $U$ and running the chain for $B$ steps
produces a terminal distribution $\mu_B$ satisfying
\[
    d_{TV}(\mu_B,p)\le \frac{\eps}{2}.
\]

The learner independently restarts this simulated chain $M$ times, each time
running it for $B$ steps, and records the terminal states. These terminal
states are i.i.d. samples from $\mu_B$. By the empirical distribution lemma (\cref{lem:empirical}),
for a sufficiently large universal constant $C$, if
\[
    M
    \ge
    C\frac{m+\log(1/\delta)}{\eps^2},
\]
then with probability at least $1-\delta$,
\[
    d_{TV}(\widehat p,\mu_B)\le \frac{\eps}{2}.
\]
On this event, the triangle inequality gives
\[
    d_{TV}(\widehat p,p)
    \le
    d_{TV}(\widehat p,\mu_B)+d_{TV}(\mu_B,p)
    \le
    \eps.
\]

Each approximate sample costs $B$ provider queries, so the total number of
queries is
\[
    MB
    =
    O\!\left(
        \frac{
            m(m+\log(1/\delta))\log(2/\eps)
        }{\eps^2}
    \right).
\]
Thus $\sellers$ PAC learns $\mathcal P_U$ whenever $\co(\sellers,U)$ is
complete.
\end{proof}

\paragraph{Why the quadratic dependence can be necessary.}
The previous upper bound raises a natural question: is the quadratic dependence on $|U|$ an artifact of this Markov chain construction, or is it necessary for some
$\sellers$ that induce complete co-occurrence graphs? The next theorem shows that the quadratic bound is
intrinsic in the worst case. Intuitively, this is because completeness shows that every pair co-occurs
somewhere, but it does not show how informative those co-occurrences are. The hardness arises from \emph{hub dilution}. We construct a query family where every allowed set $S \in \sellers$ contains a central ``hub'' element alongside a few ``light'' elements, ensuring the co-occurrence graph remains complete. In our hard family of distributions, the hub is assigned constant mass, while the light elements receive mass $\Theta(1/n)$. Consequently, any conditional query returns the uninformative hub except for an $O(1/n)$-fraction of the time. To observe a light element and gain information about its true mass, the learner must waste $\Theta(n)$ samples per query. Since the learner must calibrate $\Theta(n)$ such light elements, this dilution forces the overall sample complexity to $\Omega(n^2/\eps^2)$. The proof relies on Fano's inequality and coding-theoretic tools to make this intuition formal.

\begin{theorem}[Quadratic complexity]
\label{thm:quadratic-lower-main}
There exist absolute constants $c>0$ and $n_0\in\mathbb N$ such that for every
$n\ge n_0$, there exists a query family $\sellers$ on $[n]$ with complete
co-occurrence graph such that, for all $0<\eps\le c$,
\[
  q^\star_{\sellers,\mathcal P^+_{[n]}}(\eps,1/3)
  \ge
  c\,\frac{n^2}{\eps^2}.
\]
\end{theorem}

\begin{proof}
Let $n=2m+1$ and define
\[
 [n]=\{0\}\cup\{a_1,b_1,\ldots,a_m,b_m\}.
\]
Define the allowed query family
\[
 \cS=\{S_{jk}:1\le j<k\le m\},
 \qquad
 S_{jk}=\{0,a_j,b_j,a_k,b_k\}.
\]
The co-occurrence graph is complete: vertices from different pairs co-occur in the corresponding $S_{jk}$, vertices within the same pair co-occur in any $S_{jk}$ with $k\ne j$, and the hub $0$ co-occurs with every other vertex.

For $\sigma\in\bits^m$ and $\alpha\in(0,1/2]$, define
\[
 p^\sigma(0)=\frac12,
 \qquad
 p^\sigma(a_j)=\frac{1+\alpha\sigma_j}{4m},
 \qquad
 p^\sigma(b_j)=\frac{1-\alpha\sigma_j}{4m}.
\]
Let $\cC\subseteq\bits^m$ be the code from \Cref{lem:code}. If $\sigma,\tau\in\cC$ are distinct, then each differing coordinate contributes $\alpha/(2m)$ to TV distance, so
\[
 \TV(p^\sigma,p^\tau)
 =
 \frac{\alpha}{2m}d_H(\sigma,\tau)
 \ge \frac{\alpha}{8}.
\]
Set $\alpha=32\eps$, assuming $\eps\le1/64$. Then, the code distributions are $4\eps$-separated.

Fix a query $S_{jk}$. Let $Q_{jk}$ be the baseline distribution on $S_{jk}$ with
\[
 Q_{jk}(0)=\frac{m}{m+2},
 \qquad
 Q_{jk}(a_j)=Q_{jk}(b_j)=Q_{jk}(a_k)=Q_{jk}(b_k)=\frac{1}{2(m+2)}.
\]
The conditional law under $p^\sigma$ is
\[
 P^\sigma_{jk}(0)=\frac{m}{m+2},
\]
\[
 P^\sigma_{jk}(a_j)=\frac{1+\alpha\sigma_j}{2(m+2)},\quad
 P^\sigma_{jk}(b_j)=\frac{1-\alpha\sigma_j}{2(m+2)},
\]
and similarly for the pair $k$. Therefore
\[
 \KL(P^\sigma_{jk}\|Q_{jk})
 =
 \frac{(1+\alpha)\log(1+\alpha)+(1-\alpha)\log(1-\alpha)}{m+2}.
\]
For $|\alpha|\le1/2$, the numerator is at most $2\alpha^2$, hence
\[
 \KL(P^\sigma_{jk}\|Q_{jk})\le \frac{2\alpha^2}{m+2}.
\]

Let $\Sigma$ be uniform on $\cC$ and let $\Pi$ be the transcript of any adaptive $q$-query learner. Conditional on any history, the next query is some fixed $S_{jk}$. By \Cref{lem:mi-ref}, the conditional mutual information between $\Sigma$ and the next response is at most $2\alpha^2/(m+2)$. Therefore,
\[
 I(\Sigma;\Pi)\le \frac{2q\alpha^2}{m+2}.
\]

A learner that succeeds to TV error $\eps$ with probability at least $2/3$ permits nearest-neighbor decoding of $\Sigma$ with error probability at most $1/3$. By Fano,
\[
 I(\Sigma;\Pi)\ge c_0 m
\]
for a universal $c_0>0$. Combining the two bounds yields
\[
 q\ge c_1\frac{m(m+2)}{\alpha^2}
 =
 \Omega\!\left(\frac{n^2}{\eps^2}\right).
\]

\begin{figure}[ht]
\centering
\begin{tikzpicture}

    \coordinate (H)  at (0, 0);
    \coordinate (aj) at (-0.8, 3);
    \coordinate (bj) at (0.8, 3);
    \coordinate (ak) at (-3, -1);
    \coordinate (bk) at (-1.5, -2.5);
    \coordinate (al) at (3, -1);
    \coordinate (bl) at (1.5, -2.5);

    \draw[line width=45pt, color=yellow!45, opacity=0.7, line cap=round, line join=round] 
        (aj) -- (bj) -- (H) -- (bk) -- (ak) -- (aj);

    \draw[line width=45pt, color=magenta!25, opacity=0.7, line cap=round, line join=round] 
        (aj) -- (bj) -- (H) -- (bl) -- (al) -- (aj);

    \node[circle, fill=blue!60, text=white, minimum size=12mm, font=\large\bfseries] at (H) {$0$};
    \node[font=\footnotesize\bfseries, text=blue!80!black] at (0, -0.9) {Hub (Mass 1/2)};

    \node[circle, fill=orange!80, text=white, minimum size=9mm, font=\bfseries] at (aj) {$a_j$};
    \node[circle, fill=orange!80, text=white, minimum size=9mm, font=\bfseries] at (bj) {$b_j$};
    \node[font=\footnotesize\bfseries, text=orange!90!black] at (0, 3.8) {Pair $j$};

    \node[circle, fill=teal!70, text=white, minimum size=9mm, font=\bfseries] at (ak) {$a_k$};
    \node[circle, fill=teal!70, text=white, minimum size=9mm, font=\bfseries] at (bk) {$b_k$};
    \node[font=\footnotesize\bfseries, text=teal!90!black] at (-3, -2.2) {Pair $k$};

    \node[circle, fill=purple!60, text=white, minimum size=9mm, font=\bfseries] at (al) {$a_\ell$};
    \node[circle, fill=purple!60, text=white, minimum size=9mm, font=\bfseries] at (bl) {$b_\ell$};
    \node[font=\footnotesize\bfseries, text=purple!90!black] at (3, -2.2) {Pair $\ell$};

    \node[font=\bfseries, text=orange!80!black] at (-3, 2.5) {Query Set $S_{jk}$};
    \node[font=\bfseries, text=magenta!80!black] at (3, 2.5) {Query Set $S_{j\ell}$};

\end{tikzpicture}
\caption{Visualization of the complete co-occurrence lower bound instance. Solid, brightly colored regions represent the query sets $S_{jk}$ and $S_{j\ell}$, which intersect exactly at the hub $0$ and pair $j$.}
\label{fig:cute_solid_lower_bound}
\end{figure}
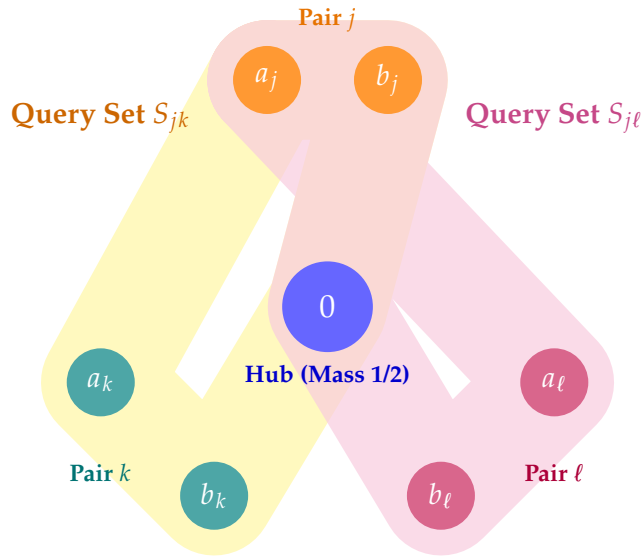
\end{proof}

\subsection{Faster Query Families: Hierarchical Comparability}\label{sec:hierarchical}
Having established the range in which the optimal sample complexity lies, 
we seek to understand for which \(\mathcal S\) the near-linear rate is attainable. Our main contribution is a general sufficient
condition that enables this rate and is based on the ability
to recursively decompose the target support into sets that are (implicitly) queryable. The following definitions are useful.

\begin{definition}[Balanced binary partition tree]
A \emph{binary partition tree} $\mathcal{T}$ over a finite set $U$ is a rooted binary tree whose nodes are nonempty subsets of $U$. Specifically, the root is $U$, the leaves are singletons $\{x\}$ for $x \in U$, and each internal node $C$ is partitioned by its two children $C_L, C_R$ (i.e., $C_L \cup C_R = C$ and $C_L \cap C_R = \emptyset$). We say $\mathcal{T}$ is \emph{balanced} if its depth is $O(\log |U|)$.
\end{definition}

\begin{definition}[Hierarchical Comparability]
\label{def:HC}
Let \(U\subseteq[n]\), and let \(\mathcal S\subseteq 2^U\setminus\{\emptyset\}\)
be a family of allowed query sets. We say that \(\mathcal S\) is
\emph{hierarchically comparable on $U$} 
 if there exists a balanced binary
partition tree \(\mathcal T\) over \(U\) such that for every internal node
\(C\in\mathcal T\) and every unordered pair \(\{x,y\}\subseteq C\), there
exists \(W_C(\{x,y\})\in\mathcal S\) satisfying
\[
\{x,y\}\subseteq W_C(\{x,y\})\subseteq C.
\]
Equivalently, every pair of elements inside every tree cell can be compared
using a query set that stays inside that cell.
\end{definition}
\begin{remark}
\label{rem:hc-examples}
We note that this property depends solely on $U, \sellers$ and not the target distribution. Many natural query families satisfy it, including
pairwise queries, hierarchical category menus, and geometric families such as
intervals, boxes, and subcubes. We provide more concrete examples in
\Cref{sec:hc-examples}.
\end{remark}

\begin{theorem}[Near-linear learning from tree-local completeness]
\label{thm:tree-local-main}
Let $U\subseteq[n]$ be nonempty with $|U| = m$ and suppose $\sellers$ is
hierarchically comparable on $U$. Then, for an absolute constant $C$,  $\sellers$ PAC learns $\mathcal P_U$ with
query complexity
\[
 q^*_{\sellers,\mathcal P_U}(\eps,\delta)
 \le
  C\,\frac{m(\log m)^2\log(m/\delta)}{\eps^2}.
\]
\end{theorem}

The following corollary follows immediately from \cref{thm:tree-local-main}.
\begin{corollary}[Pairwise queries have a nearly linear PAC rate]
\label{thm:pairwise-main}
For the pairwise query family $\sellers = \binom{[n]}{2}$ on any nonempty support $U \subseteq [n]$ of size $m$, the PAC query complexity satisfies $q^*_{\sellers,\mathcal P_U}(\eps,\delta) \le C m(\log m)^2\log(m/\delta) / \eps^2$ for an absolute constant $C$.
\end{corollary}

To get some intuition for the proof, fix a
balanced binary tree over the domain. Each internal node $C$ is split into two
children $C_L$ and $C_R$, and the goal at that node is to estimate the fraction
of mass in the left child,
$
    p(C_L\mid C).
$
Once these split probabilities are known for all internal nodes, the 
distribution is obtained by multiplying the corresponding split probabilities
along each root-to-leaf path.
The main challenge is that the learner cannot directly sample from
$p(\cdot\mid C)$. Fortunately, due to \cref{def:HC} it turns out we can simulate such local sampling
indirectly. Inside each node $C$, the learner runs a local random walk by querying an appropriate set $S \in \sellers$. This walk has stationary
distribution $p(\cdot\mid C)$, so the fraction of time it spends in $C_L$
estimates the desired split. The bulk of the technical work is to show that due to the nature of the estimation problem we do not need to restart this chain every time we obtain a sample; a Bernstein bound establishes the estimation error. 
Repeating this across all nodes of the tree gives
a recursive estimator for the entire distribution. The algorithm is presented in \Cref{alg:tree_pairwise} and \Cref{sec:linear-rates-proof} is devoted to the proof of this result.

\begin{algorithm}[ht]
\SetAlgoLined
\DontPrintSemicolon
\KwIn{Ambient set $U$ of size $m$, a supply network $\sellers$ hierarchically
comparable on $U$, target accuracy $\eps>0$, confidence $\delta>0$.}
\KwData{A witnessing balanced binary partition tree $\cT$ over $U$; witnesses
$W_C(\{x,y\})$ for all internal nodes $C$ and pairs $\{x,y\}\subseteq C$;
node error budget $\tau=\Theta(\eps^2/D)$; node confidence
$\delta_C=\delta/(2m)$.}
\For{\text{each internal node } $C\in\cT$}{
    Let $C_L,C_R$ be the children of $C$, and let $k=|C|$\;
    $B_C\gets \lceil (k-1)\log(4/\delta_C)\rceil$ \tcc*{Burn-in steps}
    $M_C\gets \Theta\!\left(k\log(4/\delta_C)/\tau\right)$
    \tcc*{Estimation steps}
    Choose an arbitrary starting state $x\in C$\;

    \tcc{Phase 1: Burn-in}
    \For{$t=1$ \KwTo $B_C$}{
        Draw challenger $y$ uniformly from $C\setminus\{x\}$\;
        Query $W_C(\{x,y\})$ and receive $z\sim Q_{p,W_C(\{x,y\})}$\;
        \lIf{$z=y$}{$x\gets y$}
    }

    \tcc{Phase 2: Estimate the left-child occupation}
    $hits\gets 0$\;
    \For{$t=1$ \KwTo $M_C$}{
        Draw challenger $y$ uniformly from $C\setminus\{x\}$\;
        Query $W_C(\{x,y\})$ and receive $z\sim Q_{p,W_C(\{x,y\})}$\;
        \lIf{$z=y$}{$x\gets y$}
        \lIf{$x\in C_L$}{$hits\gets hits+1$}
    }

    $\overline\alpha_C\gets hits/M_C$\;
    $\widehat\alpha_C
    \gets
    \min\left\{
        1-\frac{\tau}{4},
        \max\left\{
            \frac{\tau}{4},
            \overline\alpha_C
        \right\}
    \right\}$\;
}

\tcc{Reconstruct full distribution via path products}
\For{\text{each leaf } $v\in U$}{
    $\widehat p_v\gets 1$\;
    \For{\text{each internal node } $C$ \text{ on the root-to-}$v$
    \text{ path}}{
        \eIf{$v\in C_L$}{
            $\widehat p_v\gets \widehat p_v\cdot \widehat\alpha_C$\;
        }{
            $\widehat p_v\gets \widehat p_v\cdot (1-\widehat\alpha_C)$\;
        }
    }
}
\Return reconstructed distribution $\widehat p$\;
\caption{Tree-Based Recursive Learner under Hierarchical Comparability}
\label{alg:tree_pairwise}
\end{algorithm}

\subsubsection{Proof of \Cref{thm:tree-local-main}}\label{sec:linear-rates-proof}

The proof works for the support
class $\mathcal P_U$, so the target may have zero-mass points inside $U$.

Let $m:=|U|$. 
If $m=1$, the learner outputs the unique point mass, so assume $m\ge2$.
Fix a balanced binary partition tree $\cT$ over $U$ witnessing
hierarchical comparability. Let $D$ denote the depth of $\cT$, so
$D=O(\log m)$. For every internal node $C\in\cT$ and every unordered pair
$\{x,y\}\subseteq C$, fix a witness set
\[
    W_C(\{x,y\})\in\sellers
    \qquad\text{such that}\qquad
    \{x,y\}\subseteq W_C(\{x,y\})\subseteq C .
\]
The witness is symmetric in $x,y$ because it is indexed by the unordered pair
$\{x,y\}$.

For a node $C$ with $p(C)>0$, define
\[
    p_C(x):=p(x\mid C)=\frac{p_x}{p(C)},
    \qquad x\in C,
\]
and define the split parameter
\[
    \alpha_C:=p(C_L\mid C)=\frac{p(C_L)}{p(C)}.
\]
If $p(C)=0$, the value of $\alpha_C$ is irrelevant and may be chosen
arbitrarily. Such nodes contribute zero mass to the error decomposition below.

\begin{lemma}[Tree KL chain rule with zero-mass nodes]
\label{lem:tree-kl}
Suppose that for every internal node $C$ we have an estimate
$\widehat\alpha_C\in(0,1)$ and reconstruct a leaf distribution $\widehat p$
by multiplying the estimated split probabilities along each root-to-leaf path.
Then
\[
    D_{\mathrm{KL}}(p\|\widehat p)
    =
    \sum_{C\in\mathcal I(\cT):\,p(C)>0}
    p(C)\,
    D_{\mathrm{KL}}\!\left(
        \mathrm{Bern}(\alpha_C)
        \,\middle\|\,
        \mathrm{Bern}(\widehat\alpha_C)
    \right),
\]
where $\mathcal I(\cT)$ denotes the set of internal nodes of $\cT$.
Consequently,
\[
    d_{TV}(p,\widehat p)
    \le
    \sqrt{
        \frac12
        \sum_{C\in\mathcal I(\cT):\,p(C)>0}
        p(C)\,
        D_{\mathrm{KL}}\!\left(
            \mathrm{Bern}(\alpha_C)
            \,\middle\|\,
            \mathrm{Bern}(\widehat\alpha_C)
        \right)
    }.
\]
\end{lemma}

\begin{proof}
For a leaf $x$ with $p_x>0$, all nodes on the root-to-$x$ path have positive
$p$-mass. Write
\[
    p_x=\prod_{C\ni x}\beta_C(x),
    \qquad
    \widehat p_x=\prod_{C\ni x}\widehat\beta_C(x),
\]
where $\beta_C(x)=\alpha_C$ if $x\in C_L$ and
$\beta_C(x)=1-\alpha_C$ if $x\in C_R$, and similarly for
$\widehat\beta_C(x)$ using $\widehat\alpha_C$. Hence
\[
    \log\frac{p_x}{\widehat p_x}
    =
    \sum_{C\ni x}
    \log\frac{\beta_C(x)}{\widehat\beta_C(x)}.
\]
Taking expectation under $p$ and exchanging the finite sums gives
\[
\begin{aligned}
    D_{\mathrm{KL}}(p\|\widehat p)
    &=
    \sum_{C:\,p(C)>0}
    \sum_{x\in C}
    p_x
    \log\frac{\beta_C(x)}{\widehat\beta_C(x)}  \\
    &=
    \sum_{C:\,p(C)>0}
    p(C)
    \left[
        \alpha_C\log\frac{\alpha_C}{\widehat\alpha_C}
        +
        (1-\alpha_C)\log\frac{1-\alpha_C}{1-\widehat\alpha_C}
    \right].
\end{aligned}
\]
This is the desired identity. Pinsker's inequality gives the total-variation
bound.
\end{proof}

\paragraph{A local chain for estimating a split.} Fix an internal node $C$ with $p(C)>0$ and $k:=|C|\ge2$. We define a Markov
chain $K_C$ on the state space $C$. From current state $x\in C$, choose a
challenger $y$ uniformly from $C\setminus\{x\}$, query the tree-local witness
set  $W_C(\{x,y\}),$
and move to $y$ if and only if the oracle response is $y$. Otherwise, stay at
$x$.

Equivalently, for $x\neq y$,
\[
    K_C(x,y)
    =
    \frac{1}{k-1}Q_{p,W_C(\{x,y\})}(y),
\]
where $Q_{p,W_C(\{x,y\})}$ is the one-query response law from the
preliminaries. This definition is valid even when
$p(W_C(\{x,y\}))=0$.

If $k=2$, then the unique challenger is the other point of $C$, and a
direct calculation gives
$
    K_C(x,\cdot)=p_C(\cdot)$
 for every $x\in C.
$
Hence the mixing and spectral-gap conclusions hold immediately. Assume below
that $k>2$.

\begin{lemma}[Stationarity and minorization of the local chain]
\label{lem:local-chain-pu}
For every positive-mass node $C$, the conditional distribution $p_C$ is
stationary and reversible for $K_C$. Moreover, for every $x\in C$,
\[
    K_C(x,\cdot)\ge \frac{1}{k-1}p_C(\cdot)
\]
pointwise on $C$. Consequently,
\[
    d_{TV}(K_C^t(x,\cdot),p_C)
    \le
    e^{-t/(k-1)}
\]
for every starting state $x\in C$, and the spectral gap of $K_C$ is
at least $1/(k-1)$.
\end{lemma}

\begin{proof}
We first prove reversibility. Fix distinct $x,y\in C$ and let
\[
    W:=W_C(\{x,y\}).
\]
By hierarchical comparability,
\[
    x,y\in W\subseteq C.
\]

If $p(W)=0$, then since $x,y\in W$ and $p$ is nonnegative, we have
$p_x=p_y=0$. Therefore
\[
    p_C(x)K_C(x,y)=p_C(y)K_C(y,x)=0.
\]

If $p(W)>0$, then the zero-mass convention is irrelevant and
\[
    Q_{p,W}(y)=\frac{p_y}{p(W)},
    \qquad
    Q_{p,W}(x)=\frac{p_x}{p(W)}.
\]
Hence
\[
\begin{aligned}
    p_C(x)K_C(x,y)
    &=
    \frac{p_x}{p(C)}
    \frac{1}{k-1}
    \frac{p_y}{p(W)}  \\
    &=
    \frac{p_y}{p(C)}
    \frac{1}{k-1}
    \frac{p_x}{p(W)}
    =
    p_C(y)K_C(y,x).
\end{aligned}
\]
Thus detailed balance holds for every pair $x\neq y$, and so $p_C$ is
stationary and reversible for $K_C$.

We next prove the minorization. Fix $x\in C$.

First consider $y\neq x$ and again write $W:=W_C(\{x,y\})$. If $p(W)>0$, then
\[
    K_C(x,y)
    =
    \frac{1}{k-1}\frac{p_y}{p(W)}
    \ge
    \frac{p_y}{(k-1)p(C)}
    =
    \frac{p_C(y)}{k-1},
\]
because $W\subseteq C$ implies $p(W)\le p(C)$. If $p(W)=0$, then $p_y=0$, so
the same inequality holds trivially.

It remains to check the diagonal coordinate. If $p_x=0$, then
\[
    K_C(x,x)\ge0=\frac{p_C(x)}{k-1}.
\]
If $p_x>0$, then for every challenger $y\neq x$, the witness set
$W_C(\{x,y\})$ has positive mass because it contains $x$. Therefore
\[
    1-Q_{p,W_C(\{x,y\})}(y)
    \ge
    Q_{p,W_C(\{x,y\})}(x)
    =
    \frac{p_x}{p(W_C(\{x,y\}))}
    \ge
    \frac{p_x}{p(C)}
    =
    p_C(x),
\]
where the last inequality uses $W_C(\{x,y\})\subseteq C$. Averaging over
$y\in C\setminus\{x\}$ gives
\[
    K_C(x,x)
    =
    \frac{1}{k-1}
    \sum_{y\neq x}
    \bigl(1-Q_{p,W_C(\{x,y\})}(y)\bigr)
    \ge
    p_C(x)
    \ge
    \frac{p_C(x)}{k-1}.
\]
Thus the pointwise minorization holds.

The minorization yields a Doeblin decomposition
\[
    K_C
    =
    \frac{1}{k-1}\Pi_C
    +
    \left(1-\frac{1}{k-1}\right)R_C,
\]
where $\Pi_C$ is the rank-one kernel whose every row is $p_C$, and $R_C$ is a
Markov kernel. Hence
\[
    d_{TV}(K_C^t(x,\cdot),p_C)
    \le
    \left(1-\frac{1}{k-1}\right)^t
    \le
    e^{-t/(k-1)}.
\]
Since $K_C$ is reversible, the same Doeblin contraction gives absolute
spectral gap at least $1/(k-1)$.
\end{proof}

We use the following standard Bernstein inequality for stationary reversible
chains with spectral gap; see, for example, \cite{JiangSunFan2018}.

\begin{proposition}[Bernstein inequality for the local chain]
\label{prop:mc-bernstein}
There is a universal constant $c_B>0$ such that if $(X_t)$ is stationary,
reversible, has invariant law $\pi$, and has absolute spectral gap at least
$\gamma$, then for every $0\le f\le1$ with mean $\mu$ and variance $\sigma^2$,
\[
    \Pr\!\left[
        \left|
        \frac1M\sum_{t=1}^M f(X_t)-\mu
        \right|>u
    \right]
    \le
    2\exp\!\left(
        -c_B\frac{M\gamma u^2}{\sigma^2+u}
    \right).
\]
\end{proposition}

\paragraph{One-node split estimation.} At a positive-mass node $C$, the goal is to estimate
\[
    \alpha_C=p(C_L\mid C).
\]
Run the local chain for
\[
    B_C
    :=
    \left\lceil
        (k-1)\log\frac{4}{\delta_C}
    \right\rceil
\]
burn-in steps. Then, run it for $M_C$ additional steps and define
\[
    \overline\alpha_C
    :=
    \frac1{M_C}\sum_{t=1}^{M_C}\1\{X_t\in C_L\}.
\]
For a node KL budget $\tau_C\in(0,1/16]$, clip
\[
    \widehat\alpha_C
    :=
    \min\left\{
        1-\frac{\tau_C}{4},
        \max\left\{
            \frac{\tau_C}{4},
            \overline\alpha_C
        \right\}
    \right\}.
\]

\begin{lemma}[One-node KL estimate]
\label{lem:node-kl}
There are universal constants $A,C>0$ such that the following hold. Fix a
positive-mass internal node $C$ of size $k$, a node budget
$\tau_C\in(0,1/16]$, and a failure probability $\delta_C\in(0,1)$. If
\[
    M_C
    \ge
    C\,\frac{k\log(4/\delta_C)}{\tau_C},
\]
then the clipped estimator above satisfies
\[
    \Pr\!\left[
        D_{\mathrm{KL}}\!\left(
            \mathrm{Bern}(\alpha_C)
            \,\middle\|\,
            \mathrm{Bern}(\widehat\alpha_C)
        \right)
        >
        A\tau_C
    \right]
    \le
    \delta_C.
\]
The number of queries used at node $C$ is
\[
    O\!\left(
        k\log\frac{1}{\delta_C}
        +
        \frac{k\log(4/\delta_C)}{\tau_C}
    \right).
\]
\end{lemma}

\begin{proof}
By \Cref{lem:local-chain-pu}, after $B_C$ burn-in steps the current-state law
is within $\delta_C/4$ in total variation of stationarity. Since the subsequent
path is obtained by applying a Markov kernel to the initial state, the law of
the length-$M_C$ path is also within $\delta_C/4$ of the stationary path law.

For a stationary copy of the chain, apply \Cref{prop:mc-bernstein} with
\[
    f(x)=\1\{x\in C_L\},
    \qquad
    \mu=\alpha_C,
    \qquad
    \sigma^2=\alpha_C(1-\alpha_C),
    \qquad
    \gamma\ge\frac{1}{k-1}.
\]
With $C$ sufficiently large, the displayed choice of $M_C$ implies that, with
probability at least $1-\delta_C/2$ under the stationary path law,
\[
    |\overline\alpha_C-\alpha_C|
    \le
    \frac1{64}\sqrt{\alpha_C(1-\alpha_C)\tau_C}
    +
    \frac{\tau_C}{64}.
\]
Adding the burn-in discrepancy gives this event with probability at least
$1-\delta_C$ for the actual run.

It remains to convert the additive event into Bernoulli KL control. We record
the standard argument. If
$\alpha_C\in[8\tau_C,1-8\tau_C]$, then the additive error is at most a small
constant times
$\sqrt{\alpha_C(1-\alpha_C)\tau_C}$, no clipping occurs, and
\[
    D_{\mathrm{KL}}\!\left(
        \mathrm{Bern}(\alpha_C)
        \,\middle\|\,
        \mathrm{Bern}(\widehat\alpha_C)
    \right)
    \le
    O(\tau_C).
\]
If $\alpha_C<8\tau_C$, then the additive event implies
$\overline\alpha_C\le 9\tau_C$, so after clipping
\[
    \widehat\alpha_C\in[\tau_C/4,9\tau_C].
\]
Both Bernoulli parameters are then $O(\tau_C)$, and the Bernoulli KL is again
$O(\tau_C)$. The case $\alpha_C>1-8\tau_C$ is symmetric. Taking $A$ large
enough proves the lemma.
\end{proof}

\paragraph{The recursive estimator.} For a node $C$ at level $\ell$, set
\[
    \tau_C
    :=
    \frac{\eps^2}{4AD},
\]
where $A$ is the constant from \Cref{lem:node-kl}, and 
\[
    \delta_C:=\frac{\delta}{2m}.
\]
The learner runs the one-node procedure at every internal node of the tree.
If a node has zero mass, the output split is irrelevant for the analysis; the
algorithm need not know which nodes have zero mass. Reconstruct
$\widehat p$ from the estimated splits by multiplying estimated split
probabilities along root-to-leaf paths.

By \Cref{lem:node-kl} and a union bound over at most $m-1$ internal nodes, with
probability at least $1-\delta$, every positive-mass internal node $C$ satisfies
\[
    D_{\mathrm{KL}}\!\left(
        \mathrm{Bern}(\alpha_C)
        \,\middle\|\,
        \mathrm{Bern}(\widehat\alpha_C)
    \right)
    \le
    A\tau_C
    =
    \frac{\eps^2}{4D}.
\]
On this event, \Cref{lem:tree-kl} gives
\[
\begin{aligned}
    D_{\mathrm{KL}}(p\|\widehat p)
    &\le
    \sum_{\ell=0}^{D-1}
    \sum_{\substack{C\text{ at level }\ell\\ p(C)>0}}
    p(C)\,\frac{\eps^2}{4D}  \\
    &\le
    \sum_{\ell=0}^{D-1}
    \frac{\eps^2}{4D}
    =
    \frac{\eps^2}{4},
\end{aligned}
\]
because the internal nodes at a fixed level are pairwise disjoint, so their
total mass is at most one.
Pinsker's inequality yields
\[
    d_{TV}(p,\widehat p)
    \le
    \sqrt{\frac12D_{\mathrm{KL}}(p\|\widehat p)}
    \le
    \eps.
\]

It remains to count queries. The internal nodes at a fixed level are pairwise disjoint, so the sum of
their cardinalities is at most $m$. Since
\[
    \tau_C=\frac{\eps^2}{4AD}
\]
for all nodes, the total provider-query cost on one level is
\[
    O\!\left(
        \frac{m\log(m/\delta)}{\tau_C}
    \right)
    =
    O\!\left(
        \frac{mD\log(m/\delta)}{\eps^2}
    \right).
\]
Summing over the $D=O(\log m)$ levels gives
\[
    O\!\left(
        \frac{mD^2\log(m/\delta)}{\eps^2}
    \right)
    =
    O\!\left(
        \frac{m(\log m)^2\log(m/\delta)}{\eps^2}
    \right).
\]
This proves \Cref{thm:tree-local-main}.

\subsection{A Continuum of Optimal Sample Complexities}\label{sec:continuum}
So far we have shown that whenever PAC learning is possible the sample
complexity cannot be worse than $\tilde{O}(n^2/\epsilon^2)$. Moreover,
there are families $\sellers$ for which $\Omega(n^2/\epsilon^2)$ samples are necessary, and for more well-structured $\sellers$ the learner can achieve $\tilde{\Theta}(n/\epsilon^2)$ sample complexity.
Thus, it is natural to ask if intermediate sample complexities of the form $\tilde{\Theta}(n^{1+\gamma}/\epsilon^2)$
can arise as optimal bounds for certain $\sellers.$ Our next result shows that this is indeed the case. The proof relies on a generalization 
of our quadratic lower bound construction and is deferred to \Cref{app:intermediate}.

\begin{theorem}[Continuum of optimal polynomial rates]
\label{thm:intermediate-rates}
There exist universal constants $c_1, c_2, \varepsilon_0 > 0$ such that for every exponent $\gamma \in (0,1)$, there exists some $N_\gamma$ where for all $n \geq N_\gamma$, one can construct a query family $\sellers$ on $[n]$ satisfying
\[
    c_1 \frac{n^{1+\gamma}}{\varepsilon^2}
    \leq q^\star_{\sellers,\mathcal P_{[n]}}(\varepsilon, 1/3) \leq
    c_2 \frac{n^{1+\gamma} \log(1/\varepsilon) \log n}{\varepsilon^2}
\]
for all $0 < \varepsilon \le \varepsilon_0$. In particular, this yields an optimal sample complexity of
\[
    q^\star_{\sellers,\mathcal P_{[n]}}(\varepsilon, 1/3) = \widetilde{\Theta}\left( \frac{n^{1+\gamma}}{\varepsilon^2} \right).
\]
\end{theorem}

\begin{remark}[The structural parameter behind the intermediate rates]
\label{rem:intermediate-parameter}
The construction in \Cref{thm:intermediate-rates} has a natural two-scale
structure. Contract each internally queryable block $A_j$ and $B_j$ to a
single coarse atom, while keeping the hub $0$ as a singleton, and let $H$
denote the number of atoms in the resulting coarse domain. Thus, $H$ is the
size of the quotient obtained after contracting the fine blocks.

Learning the lifted distribution consists of two statistical tasks. First, the
learner must recover the conditional distributions within the blocks, whose
total dimension is $n$ and whose cost is
$\widetilde\Theta(n/\eps^2)$. Second, it must recover the masses of the $H$
coarse atoms. The induced query family on these atoms is the hub-diluted hard
family, whose cost is $\widetilde\Theta(H^2/\eps^2)$. Consequently, the
block-lifted family has optimal rate
\[
    \widetilde\Theta\left(\frac{n+H^2}{\eps^2}\right).
\]

For the family in \Cref{thm:intermediate-rates}, we choose $H\asymp n^{(1+\gamma)/2},$ which proves the result.
\end{remark}

\subsection{Distribution-Dependent Rates for Pointwise Consistency}\label{sec:pointwise-rates}
We conclude this section by providing a bound on the pointwise-consistency rates of our task. Our result here shows that whenever learning is possible, the error rate decays as $C/\sqrt{t},$ for some distribution-dependent constant $C$, where $t$ is the number of samples the learner is using.

\begin{theorem}[Distribution-dependent pointwise rate]
\label{thm:pointwise-rate}
Let $U\subseteq[n]$ be nonempty and suppose that
$\co(\sellers,U)$ is connected. There exists a single sequence of learners
$(A_t)_{t\ge 1}$ such that, for every target
$p\in\mathcal P_U^+$, there is a finite constant
$C_{p,\sellers,U}>0$ satisfying the following: for every
$\eps\in(0,1)$ and $\delta\in(0,1/2)$,
\[
    t
    \ge
    C_{p,\sellers,U}
    \frac{\log(1/\delta)}{\eps^2}
\]
implies
\[
    \Prob_p\!\left[\TV(A_t,p)\le \eps\right]
    \ge 1-\delta.
\]
\end{theorem}

\begin{proof}
We use the spanning-tree learner from the proof of
\Cref{thm:market_pointwise}. The case $|U|=1$ is immediate, so let
$m:=|U|\ge2$. Fix a rooted spanning tree $\mathcal T$ of
$\co(\sellers,U)$ and a witness set $S_e\in\sellers$ for every directed tree
edge $e=(u\to v)$.

For an endpoint $y\in\{u,v\}$ of an edge $e$, define
\[
    \theta_e(y):=p(y\mid S_e).
\]
Because $p\in\mathcal P_U^+$, every one of these quantities is strictly
positive. Collect the $2(m-1)$ endpoint probabilities into a vector
$\theta(p)$.

Let $F$ denote the deterministic reconstruction map used in the spanning-tree
learner: given positive endpoint probabilities, it forms the corresponding
edge ratios, multiplies them along root-to-vertex paths, and normalizes the
resulting vector. By construction,
\[
    F(\theta(p))=p.
\]
The map $F$ is continuously differentiable on the positive orthant. Therefore,
there exist constants $r_p>0$ and $L_p<\infty$ such that
\[
    \TV(F(z),p)
    \le
    L_p\|z-\theta(p)\|_\infty
\]
whenever
\[
    \|z-\theta(p)\|_\infty\le r_p.
\]

Given a total budget $t$, the learner uses
\[
    k_t:=\left\lfloor\frac{t}{m-1}\right\rfloor
\]
samples for every tree edge. Let $\widehat\theta^{(t)}$ be the vector of
empirical endpoint frequencies. Each coordinate is the empirical mean of
$k_t$ independent Bernoulli variables. Hence, by Hoeffding's inequality and a
union bound, for every $u>0$,
\[
    \Prob_p\!\left[
        \|\widehat\theta^{(t)}-\theta(p)\|_\infty>u
    \right]
    \le
    4(m-1)e^{-2k_tu^2}.
\]

The learner uses the same lower truncation as in the proof of
\Cref{thm:market_pointwise}. Since all coordinates of $\theta(p)$ are
positive, for all sufficiently large $k_t$---where the required threshold may
depend on $p$---this truncation is inactive whenever
$\widehat\theta^{(t)}$ lies in a sufficiently small neighborhood of
$\theta(p)$.

Let
\[
    u:=\min\left\{r_p,\frac{\eps}{L_p}\right\}.
\]
On the event
\[
    \|\widehat\theta^{(t)}-\theta(p)\|_\infty\le u,
\]
the reconstructed distribution therefore satisfies
\[
    \TV(A_t,p)
    \le L_pu
    \le\eps.
\]
It is enough to take
\[
    k_t
    \ge
    \frac{1}{2u^2}
    \log\frac{4(m-1)}{\delta}.
\]
Since $p$, $U$, and the chosen spanning tree are fixed,
\[
    \frac{1}{u^2}
    \le
    \frac{C'_p}{\eps^2}
\]
for some finite constant $C'_p$, for every $\eps\in(0,1)$. Absorbing the
number of tree edges, the fixed truncation threshold, and
$\log(4(m-1))$ into a constant $C_{p,\sellers,U}$ proves the result.
\end{proof}

\section{Conclusion}\label{sec:conclusion}
In this work, we studied the problem of learning discrete distributions from arbitrary restricted conditional query families $\sellers$. Our findings provide a complete qualitative characterization of learnability, as well as the range of the optimal sample complexity bounds for different $\sellers.$ 
This framework opens several promising avenues for future research. First, while our sample complexity bounds are tight up to logarithmic factors, closing the remaining polylogarithmic gaps is an immediate next step. Second, our current analysis focuses strictly on distributions with finite support; extending these structural insights to more general spaces, such as countably infinite or continuous domains, presents a natural and compelling theoretical challenge. 
Finally, it would be interesting to explore the graph-theoretic conditions that $\sellers$ must satisfy to enable other fundamental statistical tasks, such as distribution testing or property estimation from arbitrary $\sellers$.

\newpage

\bibliographystyle{alpha}
\bibliography{refs}

@article{NOS17,
author = {Negahban, Sahand and Oh, Sewoong and Shah, Devavrat},
title = {Rank Centrality: Ranking from Pairwise Comparisons},
journal = {Operations Research},
volume = {65},
number = {1},
pages = {266-287},
year = {2017},
doi = {10.1287/opre.2016.1534},

URL = { 
    
        https://doi.org/10.1287/opre.2016.1534
    
    

},
eprint = { 
    
        https://doi.org/10.1287/opre.2016.1534
    
    

}
}

@article{CanonneRonServedio2015,
author = {Canonne, Cl\'{e}ment L. and Ron, Dana and Servedio, Rocco A.},
title = {Testing Probability Distributions using Conditional Samples},
journal = {SIAM Journal on Computing},
volume = {44},
number = {3},
pages = {540-616},
year = {2015},
doi = {10.1137/130945508},

URL = { 
    
        https://doi.org/10.1137/130945508
    
    

},
eprint = { 
    
        https://doi.org/10.1137/130945508
    
    

}
}

@article{ChakrabortyFischerGoldhirshMatsliah2016,
author = {Chakraborty, Sourav and Fischer, Eldar and Goldhirsh, Yonatan and Matsliah, Arie},
title = {On the Power of Conditional Samples in Distribution Testing},
journal = {SIAM Journal on Computing},
volume = {45},
number = {4},
pages = {1261-1296},
year = {2016},
doi = {10.1137/140964199},

URL = { 
    
        https://doi.org/10.1137/140964199
    
    

},
eprint = { 
    
        https://doi.org/10.1137/140964199
    
    

}
}

@inproceedings{KamathTzamos2019,
author = {Kamath, Gautam and Tzamos, Christos},
title = {Anaconda: a non-adaptive conditional sampling algorithm for distribution testing},
year = {2019},
publisher = {Society for Industrial and Applied Mathematics},
address = {USA},
booktitle = {Proceedings of the Thirtieth Annual ACM-SIAM Symposium on Discrete Algorithms},
pages = {679–693},
numpages = {15},
location = {San Diego, California},
series = {SODA '19}
}

@book{Canonne2020Survey,
 author = {Canonne, Cl{\'{e}}ment L.},
 title = {A Survey on Distribution Testing: Your Data is Big. But is it Blue?},
 year = {2020},
 pages = {1--100},
 doi = {10.4086/toc.gs.2020.009},
 publisher = {Theory of Computing Library},
 number = {9},
 series = {Graduate Surveys},
 URL = {http://www.theoryofcomputing.org/library.html},
}

@inproceedings{AdarFischerLevi2025,
  author       = {Tomer Adar and
                  Eldar Fischer and
                  Amit Levi},
  editor       = {Kasper Green Larsen and
                  Barna Saha},
  title        = {Optimal mass estimation in the conditional sampling model},
  booktitle    = {Proceedings of the 2026 Annual {ACM-SIAM} Symposium on Discrete Algorithms,
                  {SODA} 2026, Vancouver, BC, Canada, January 11-14, 2026},
  pages        = {4105--4174},
  publisher    = {{SIAM}},
  year         = {2026},
  url          = {https://doi.org/10.1137/1.9781611978971.152},
  doi          = {10.1137/1.9781611978971.152},
  bibsource    = {dblp computer science bibliography, https://dblp.org}
}

@book{luce1959,
  title={Individual Choice Behavior: A Theoretical Analysis},
  author={Luce, R.D.},
  lccn={59009346},
  url={https://books.google.com/books?id=a80DAQAAIAAJ},
  year={1959},
  publisher={Wiley}
}

@article{AlosFerrerMihm2025,
title = {A characterization of the Luce choice rule for an arbitrary collection of menus},
journal = {Journal of Economic Theory},
volume = {223},
pages = {105941},
year = {2025},
issn = {0022-0531},
doi = {https://doi.org/10.1016/j.jet.2024.105941},
url = {https://www.sciencedirect.com/science/article/pii/S0022053124001479},
author = {Carlos Alós-Ferrer and Maximilian Mihm}
}

@inproceedings{MaystreGrossglauser2015,
author = {Maystre, Lucas and Grossglauser, Matthias},
title = {Fast and accurate inference of Plackett-Luce models},
year = {2015},
publisher = {MIT Press},
address = {Cambridge, MA, USA},
booktitle = {Proceedings of the 29th International Conference on Neural Information Processing Systems - Volume 1},
pages = {172–180},
numpages = {9},
location = {Montreal, Canada},
series = {NIPS'15}
}

@article{roberts2004general,
   title={General state space Markov chains and MCMC algorithms},
   volume={1},
   ISSN={1549-5787},
   url={http://dx.doi.org/10.1214/154957804100000024},
   DOI={10.1214/154957804100000024},
   number={none},
   journal={Probability Surveys},
   publisher={Institute of Mathematical Statistics},
   author={Roberts, Gareth O. and Rosenthal, Jeffrey S.},
   year={2004},
   month=Jan }

@book{shalev2014understanding,
  title={Understanding machine learning: From theory to algorithms},
  author={Shalev-Shwartz, Shai and Ben-David, Shai},
  year={2014},
  publisher={Cambridge university press}
}

@article{canonne2020short,
  title={A short note on learning discrete distributions},
  author={Canonne, Cl{\'e}ment L},
  journal={arXiv preprint arXiv:2002.11457},
  year={2020}
}

@article{rosenthal1995minorization,
  title={Minorization conditions and convergence rates for Markov chain Monte Carlo},
  author={Rosenthal, Jeffrey S},
  journal={Journal of the American Statistical Association},
  volume={90},
  number={430},
  pages={558--566},
  year={1995},
  publisher={Taylor \& Francis}
}

@InProceedings{ShahBalakrishnanBradleyParekhRamchandranWainwright2016,
  title = 	 {{Estimation from Pairwise Comparisons: Sharp Minimax Bounds with Topology Dependence}},
  author = 	 {Shah, Nihar and Balakrishnan, Sivaraman and Bradley, Joseph and Parekh, Abhay and Ramchandran, Kannan and Wainwright, Martin},
  booktitle = 	 {Proceedings of the Eighteenth International Conference on Artificial Intelligence and Statistics},
  pages = 	 {856--865},
  year = 	 {2015},
  editor = 	 {Lebanon, Guy and Vishwanathan, S. V. N.},
  volume = 	 {38},
  series = 	 {Proceedings of Machine Learning Research},
  address = 	 {San Diego, California, USA},
  month = 	 {09--12 May},
  publisher =    {PMLR},
  url = 	 {https://proceedings.mlr.press/v38/shah15.html}
}

@InProceedings{AgarwalPatilAgarwal2018,
  title = 	 {Accelerated Spectral Ranking},
  author =       {Agarwal, Arpit and Patil, Prathamesh and Agarwal, Shivani},
  booktitle = 	 {Proceedings of the 35th International Conference on Machine Learning},
  pages = 	 {70--79},
  year = 	 {2018},
  editor = 	 {Dy, Jennifer and Krause, Andreas},
  volume = 	 {80},
  series = 	 {Proceedings of Machine Learning Research},
  month = 	 {10--15 Jul},
  publisher =    {PMLR},
  url = 	 {https://proceedings.mlr.press/v80/agarwal18b.html}
}

@inproceedings{FotakisKalavasisTzamos2022,
author = {Fotakis, Dimitris and Kalavasis, Alkis and Tzamos, Christos},
title = {Perfect sampling from pairwise comparisons},
year = {2022},
isbn = {9781713871088},
publisher = {Curran Associates Inc.},
address = {Red Hook, NY, USA},
booktitle = {Proceedings of the 36th International Conference on Neural Information Processing Systems},
articleno = {1857},
numpages = {16},
location = {New Orleans, LA, USA},
series = {NIPS '22}
}

@inproceedings{SahaGopalan2019,
  author       = {Aadirupa Saha and
                  Aditya Gopalan},
  editor       = {Kamalika Chaudhuri and
                  Masashi Sugiyama},
  title        = {Active Ranking with Subset-wise Preferences},
  booktitle    = {The 22nd International Conference on Artificial Intelligence and Statistics,
                  {AISTATS} 2019, 16-18 April 2019, Naha, Okinawa, Japan},
  series       = {Proceedings of Machine Learning Research},
  pages        = {3312--3321},
  publisher    = {{PMLR}},
  year         = {2019},
  url          = {http://proceedings.mlr.press/v89/saha19a.html},
  bibsource    = {dblp computer science bibliography, https://dblp.org}
}

@article{JiangSunFan2018,
  title={Bernstein's inequalities for general Markov chains},
  author={Jiang, Bai and Sun, Qiang and Fan, Jianqing},
  journal={arXiv preprint arXiv:1805.10721},
  year={2018}
}

@inproceedings{VillalobosEtAl2024,
  author       = {Pablo Villalobos and
                  Anson Ho and
                  Jaime Sevilla and
                  Tamay Besiroglu and
                  Lennart Heim and
                  Marius Hobbhahn},
  editor       = {Ruslan Salakhutdinov and
                  Zico Kolter and
                  Katherine A. Heller and
                  Adrian Weller and
                  Nuria Oliver and
                  Jonathan Scarlett and
                  Felix Berkenkamp},
  title        = {Position: Will we run out of data? Limits of {LLM} scaling based on
                  human-generated data},
  booktitle    = {Forty-first International Conference on Machine Learning, {ICML} 2024,
                  Vienna, Austria, July 21-27, 2024},
  series       = {Proceedings of Machine Learning Research},
  pages        = {49523--49544},
  publisher    = {{PMLR} / OpenReview.net},
  year         = {2024},
  url          = {https://proceedings.mlr.press/v235/villalobos24a.html},
  bibsource    = {dblp computer science bibliography, https://dblp.org}
}

@book{Bretto2013Hypergraph,
author = {Bretto, Alain},
title = {Hypergraph Theory: An Introduction},
year = {2013},
isbn = {3319000799},
publisher = {Springer Publishing Company, Incorporated}
}

@article{BradleyTerry1952,
 ISSN = {00063444, 14643510},
 URL = {http://www.jstor.org/stable/2334029},
 author = {Ralph Allan Bradley and Milton E. Terry},
 journal = {Biometrika},
 number = {3/4},
 pages = {324--345},
 publisher = {[Oxford University Press, Biometrika Trust]},
 title = {Rank Analysis of Incomplete Block Designs: I. The Method of Paired Comparisons},
 urldate = {2026-05-07},
 volume = {39},
 year = {1952}
}

@inproceedings{
xie2023doremi,
title={DoReMi: Optimizing Data Mixtures Speeds Up Language Model Pretraining},
author={Sang Michael Xie and Hieu Pham and Xuanyi Dong and Nan Du and Hanxiao Liu and Yifeng Lu and Percy Liang and Quoc V Le and Tengyu Ma and Adams Wei Yu},
booktitle={Thirty-seventh Conference on Neural Information Processing Systems},
year={2023},
url={https://openreview.net/forum?id=lXuByUeHhd}
}

@inproceedings{
ye2025datamixing,
title={Data Mixing Laws: Optimizing Data Mixtures by Predicting Language Modeling Performance},
author={Jiasheng Ye and Peiju Liu and Tianxiang Sun and Jun Zhan and Yunhua Zhou and Xipeng Qiu},
booktitle={The Thirteenth International Conference on Learning Representations},
year={2025},
url={https://openreview.net/forum?id=jjCB27TMK3}
}

@inproceedings{
lu2024dataacquisition,
title={Data Acquisition via Experimental Design for Data Markets},
author={Charles Lu and Baihe Huang and Sai Praneeth Karimireddy and Praneeth Vepakomma and Michael Jordan and Ramesh Raskar},
booktitle={The Thirty-eighth Annual Conference on Neural Information Processing Systems},
year={2024},
url={https://openreview.net/forum?id=VXJVNdmXO4}
}

@article{huang2022skill,
author = {Huang, Shiyang and Xiong, Yan and Yang, Liyan},
title = {Skill Acquisition and Data Sales},
journal = {Management Science},
volume = {68},
number = {8},
pages = {6116-6144},
year = {2022},
doi = {10.1287/mnsc.2021.4117},

URL = { 
    
        https://doi.org/10.1287/mnsc.2021.4117
    
    

},
eprint = { 
    
        https://doi.org/10.1287/mnsc.2021.4117
    
    

}
}

@article{Var85,
 ISSN = {00905364, 21688966},
 URL = {http://www.jstor.org/stable/2241152},
 author = {Y. Vardi},
 journal = {The Annals of Statistics},
 number = {1},
 pages = {178--203},
 publisher = {Institute of Mathematical Statistics},
 title = {Empirical Distributions in Selection Bias Models},
 urldate = {2026-05-07},
 volume = {13},
 year = {1985}
}

@article{GVW88,
 ISSN = {00905364, 21688966},
 URL = {http://www.jstor.org/stable/2241619},
 author = {Richard D. Gill and Yehuda Vardi and Jon A. Wellner},
 journal = {The Annals of Statistics},
 number = {3},
 pages = {1069--1112},
 publisher = {Institute of Mathematical Statistics},
 title = {Large Sample Theory of Empirical Distributions in Biased Sampling Models},
 urldate = {2026-05-07},
 volume = {16},
 year = {1988}
}

@inproceedings{BHPQ17,
 author = {Blum, Avrim and Haghtalab, Nika and Procaccia, Ariel and Qiao, Mingda},
 booktitle = {Advances in Neural Information Processing Systems},
 editor = {I. Guyon and U. Von Luxburg and S. Bengio and H. Wallach and R. Fergus and S. Vishwanathan and R. Garnett},
 pages = {},
 publisher = {Curran Associates, Inc.},
 title = {Collaborative PAC Learning},
 url = {https://proceedings.neurips.cc/paper_files/paper/2017/file/186a157b2992e7daed3677ce8e9fe40f-Paper.pdf},
 volume = {30},
 year = {2017}
}

@inproceedings{HJZ22,
 author = {Haghtalab, Nika and Jordan, Michael and Zhao, Eric},
 booktitle = {Advances in Neural Information Processing Systems},
 doi = {10.52202/068431-0030},
 editor = {S. Koyejo and S. Mohamed and A. Agarwal and D. Belgrave and K. Cho and A. Oh},
 pages = {406--419},
 publisher = {Curran Associates, Inc.},
 title = {On-Demand Sampling: Learning Optimally from Multiple Distributions},
 url = {https://proceedings.neurips.cc/paper_files/paper/2022/file/02917acec264a52a729b99d9bc857909-Paper-Conference.pdf},
 volume = {35},
 year = {2022}
}

@article{ZZC+24,
author = {Zhang, Zihan and Zhan, Wenhao and Chen, Yuxin and Du, Simon S. and Lee, Jason},
title = {Optimal Multi-Distribution Learning},
year = {2025},
issue_date = {October 2025},
publisher = {Association for Computing Machinery},
address = {New York, NY, USA},
volume = {72},
number = {5},
issn = {0004-5411},
url = {https://doi.org/10.1145/3760256},
doi = {10.1145/3760256},
journal = {J. ACM},
month = oct,
articleno = {36},
numpages = {71}
}

@InProceedings{Pen24,
  title = 	 {The sample complexity of multi-distribution learning},
  author =       {Peng, Binghui},
  booktitle = 	 {Proceedings of Thirty Seventh Conference on Learning Theory},
  pages = 	 {4185--4204},
  year = 	 {2024},
  editor = 	 {Agrawal, Shipra and Roth, Aaron},
  volume = 	 {247},
  series = 	 {Proceedings of Machine Learning Research},
  month = 	 {30 Jun--03 Jul},
  publisher =    {PMLR},
  url = 	 {https://proceedings.mlr.press/v247/peng24b.html}
}

@InProceedings{HSJ26,
  title = 	 {Is {Multi-Distribution Learning} as Easy as {PAC Learning}: Sharp Rates with Bounded Label Noise},
  author =       {Hanashiro, Rafael and Shetty, Abhishek and Jaillet, Patrick},
  booktitle = 	 {Proceedings of Thirty Ninth Conference on Learning Theory},
  pages = 	 {3109--3142},
  year = 	 {2026},
  editor = 	 {Hanneke, Steve and Lattimore, Tor},
  volume = 	 {336},
  series = 	 {Proceedings of Machine Learning Research},
  month = 	 {29 Jun--03 Jul},
  publisher =    {PMLR},
  url = 	 {https://proceedings.mlr.press/v336/hanashiro26a.html}
}

\newpage

\crefalias{section}{appendix}
\appendix

\section{Technical Lemmas}
\label{app:tech-lemmas}
Below we state several technical lemmas that are useful for our main results.
The first folklore result shows the error rate of 
a simple learner that queries $N$ times the set $[n]$.

\begin{lemma}[Empirical distribution]
\label{lem:empirical}
Let $\Omega$ be a finite set of size $m$, let $p$ be an arbitrary distribution on $\Omega$, and let $X_1,\ldots,X_N$ be independent samples from $p$.  Define the empirical distribution
\[
  \widehat p(x):=\frac1N\sum_{t=1}^N \1\{X_t=x\},
  \qquad x\in\Omega .
\]
There is a universal constant $C>0$ such that, for every $\eps,\delta\in(0,1)$, if
$
  N\ge C\frac{m+
\log(1/\delta)}{\eps^2},
$
then
$
  \Prob\bigl[\TV(\widehat p,p)\le \eps\bigr]\ge 1-\delta .
$
\end{lemma}

\begin{proof}
For two distributions $\mu,\nu$ on a finite set, by definition $
  \TV(\mu,\nu)=\sup_{A\subseteq\Omega}|\mu(A)-\nu(A)|$
holds.

It is therefore enough to control the empirical mass of every subset of $\Omega$.  Fix $A\subseteq\Omega$.  The random variables
\[
  B_t:=\1\{X_t\in A\},
  \qquad t=1,\ldots,N,
\]
are i.i.d. Bernoulli random variables with mean $p(A)$, and
$
  \widehat p(A)=\frac1N\sum_{t=1}^N B_t .
$
By the Bernoulli tail bound stated above,
\[
  \Prob\bigl[|\widehat p(A)-p(A)|>\eps\bigr]
  \le 2e^{-2N\eps^2} .
\]
There are at most $2^m$ subsets $A\subseteq\Omega$, so a union bound gives
\[
  \Prob\bigl[\TV(\widehat p,p)>\eps\bigr]
  =\Prob\left[\sup_{A\subseteq\Omega}|\widehat p(A)-p(A)|>\eps\right]
  \le 2^{m+1}e^{-2N\eps^2} .
\]
The right-hand side is at most $\delta$ whenever
\[
  N\ge \frac{(m+1)\log 2+
  \log(1/\delta)}{2\eps^2} .
\]
By choosing the appropriate absolute constant $C$ in the statement we get the claimed sufficient condition.
\end{proof}

The next result employs standard ideas
from coding theory to show that there is 
a large set of words that are pairwise far enough in Hamming distance. For completeness, we provide a proof.

\begin{lemma}[Large separated code]
\label{lem:code}
There are universal constants $c_0>0$ and $m_0\in\mathbb N$ such that the following holds for every integer $m\ge m_0$.  There exists a set $\cC\subseteq\bits^m$ of sign vectors with
$
  \log|\cC|\ge c_0 m
$
and such that any two distinct $\sigma,
\tau\in\cC$ have Hamming distance at least $m/4$:
\[
  d_H(\sigma,\tau):=\bigl|\{j\in[m]:\sigma_j\ne \tau_j\}\bigr|
  \ge \frac{m}{4} .
\]
\end{lemma}

\begin{proof}
Let $M=\lceil e^{m/100}\rceil$ and draw
$\Sigma^{(1)},\ldots,\Sigma^{(M)}$ independently and uniformly from
$\{-1,+1\}^m$. For each $a<b$,
\[
d_H(\Sigma^{(a)},\Sigma^{(b)})
\sim \operatorname{Bin}(m,1/2),
\]
so Hoeffding's inequality gives
\[
\Pr\!\left[
d_H(\Sigma^{(a)},\Sigma^{(b)})<\frac m4
\right]
\le e^{-m/8}.
\]
Hence the expected number of bad pairs is at most
\[
\binom M2 e^{-m/8}<1
\]
for all sufficiently large $m$. Therefore there is a realization with no
bad pair. Taking $\mathcal C$ to be that realization proves the distance
property, and
\[
\log|\mathcal C|\ge \frac{m}{100}.
\]
\end{proof}

The following result shows that if an estimator of a random variable
has a constant probability of predicting it correctly, then the mutual information between the estimator and the underlying variable admits a non-trivial lower bound.

\begin{lemma}[Fano inequality]
\label{lem:fano}
Let $V$ be uniformly distributed on a finite set $\mathcal V$ with $|\mathcal V|\ge2$, and let $\widehat V$ be any estimator of $V$, possibly randomized and possibly constructed from arbitrary data.  If
$
  \Prob[\widehat V\ne V]\le \frac13,
$
then
$
  I(V;\widehat V)
  \ge \frac23\log|\mathcal V|-\log 2 .
$
More generally, if $\Prob[\widehat V\ne V]\le p_e$, then
$
  I(V;\widehat V)
  \ge (1-p_e)\log|\mathcal V|-\log 2 .
$
\end{lemma}

\begin{proof}
Let
$
  E:=\1\{\widehat V\ne V\} .
$
Since $V$ is uniform on $\mathcal V$,
$
  H(V)=\log|\mathcal V| .
$
We upper bound the conditional entropy of $V$ given $\widehat V$.  By the chain rule for entropy,
\[
  H(V\mid \widehat V)
  \le H(E,V\mid \widehat V)
  = H(E\mid \widehat V)+H(V\mid E,
\widehat V).
\]
The first term is at most $\log2$ because $E$ is binary.  For the second term, if $E=0$, then $V=\widehat V$, so the conditional entropy is zero.  If $E=1$, the variable $V$ can take at most $|\mathcal V|$ values, so the conditional entropy is at most $\log|\mathcal V|$.  Therefore
\[
  H(V\mid E,
\widehat V)
  \le \Prob[E=1]\log|\mathcal V|
  \le p_e\log|\mathcal V| .
\]
Combining the last two displays gives
\[
  H(V\mid \widehat V)
  \le \log2+p_e\log|\mathcal V|.
\]
Thus
\[
  I(V;\widehat V)
  =H(V)-H(V\mid \widehat V)
  \ge (1-p_e)\log|\mathcal V|-\log2 .
\]
The stated $p_e=1/3$ case follows immediately.
\end{proof}

The next result employs standard ideas from information theory.

\begin{lemma}[Average KL to a reference bounds mutual information]
\label{lem:mi-ref}
Let $V$ be a random variable taking values in a finite set $\mathcal V$.  Conditional on $V=v$, let $X$ have distribution $P_v$ on a finite set $\mathcal X$.  Let $\mu(v)=\Prob[V=v]$, and let
$
  \overline P:=\sum_{v\in\mathcal V}\mu(v)P_v
$
be the marginal distribution of $X$.  Then, for every distribution $Q$ on $\mathcal X$,
\[
  I(V;X)
  \le \sum_{v\in\mathcal V}\mu(v)\KL(P_v\|Q) \,.
\]

\end{lemma}

\begin{proof}
If the right-hand side is $+\infty$, the claim is trivial, so assume it is finite.  Then every $P_v$ with $\mu(v)>0$ is absolutely continuous with respect to $Q$, and so is $\overline P$.

By definition of mutual information for finite variables,
\[
  I(V;X)=\sum_{v\in\mathcal V}\mu(v)\KL(P_v\|\overline P).
\]
For each $x\in\mathcal X$ with $\overline P(x)>0$, insert the factor $\overline P(x)/Q(x)$ into the KL divergence:
\begin{align*}
  \sum_{v}\mu(v)\KL(P_v\|Q)
  &=\sum_{v}\mu(v)\sum_x P_v(x)
    \log\frac{P_v(x)}{Q(x)} \\
  &=\sum_{v}\mu(v)\sum_x P_v(x)
    \log\frac{P_v(x)}{\overline P(x)}
    +\sum_{v}\mu(v)\sum_x P_v(x)
    \log\frac{\overline P(x)}{Q(x)} \\
  &=I(V;X)+\sum_x \overline P(x)
    \log\frac{\overline P(x)}{Q(x)} \\
  &=I(V;X)+\KL(\overline P\|Q) .
\end{align*}
The last term is nonnegative, which proves the inequality.
\end{proof}

For a nonempty query set $S\subseteq[n]$, a distribution $p\in\Delta([n])$,
and a zero-mass convention $\star\in\{\fail,\unifzero\}$, define
$Q^\star_{p,S}$ as the oracle's one-query response law on the alphabet
$S\cup\{\fail\}$ as follows. If $p(S)>0$, then
\[
  Q^\star_{p,S}(x)=\frac{p_x}{p(S)}\1\{x\in S\},
  \qquad
  Q^\star_{p,S}(\fail)=0 .
\]
If $p(S)=0$ and $\star=\fail$, then $Q^\star_{p,S}=\delta_{\fail}$.
If $p(S)=0$ and $\star=\unifzero$, then $Q^\star_{p,S}$ is uniform on $S$
and assigns mass zero to $\fail$.

For the result below, the transcript may include the internal random seed,
queried sets, oracle responses, and stopping indicator. Early stopping is
interpreted by padding the remaining rounds with deterministic dummy symbols.
This result shows that if the laws
of the oracles are close in TV distance
for every query set in $\cS,$ then
the distributions of the transcripts of any adaptive learner
must also be close.

\begin{lemma}[Adaptive transcript hybrid]
\label{lem:hybrid}
Fix a finite domain $[n]$, an allowed query family
$\cS\subseteq 2^{[n]}\setminus\{\emptyset\}$, and a zero-mass convention
$\star\in\{\fail,\unifzero\}$. Let $p,q\in\Delta([n])$ satisfy
\[
  \sup_{S\in\cS}
  \TV\bigl(Q^\star_{p,S},Q^\star_{q,S}\bigr)
  \le \alpha .
\]
Then, for every randomized adaptive strategy making at most $T$ queries from
$\cS$, the total variation distance between its full transcript law under $p$
and its full transcript law under $q$ is at most $T\alpha$.
\end{lemma}

\begin{proof}
Encode all internal randomization of the learner by a random seed $R$,
independent of the target. We first condition on $R=r$. After this conditioning,
the strategy is deterministic: at each round, the next query is a deterministic
function of the previous transcript, unless the strategy has already stopped.
If it has stopped, all later rounds are padded by a fixed dummy query-response
symbol. Thus it suffices to prove the claim for deterministic fixed-seed
strategies and then average over $R$.

For a deterministic strategy, let $H_t$ denote the padded transcript after
$t$ rounds, including the queried sets and responses. Let $P_t$ and $Q_t$ be
the laws of $H_t$ under targets $p$ and $q$, respectively. We prove by induction
that
\[
  \TV(P_t,Q_t)\le t\alpha,
  \qquad t=0,1,\dots,T .
\]
The case $t=0$ is trivial because the empty transcript has the same law under
both targets.

We use a standard one-step fact about Markov kernels. Here a Markov kernel
$K$ from a finite set $\mathcal H$ to a finite set $\mathcal H'$ simply means
that, for each current history $h\in\mathcal H$, $K(h,\cdot)$ is a probability
distribution on possible next histories in $\mathcal H'$. If $P$ is a
distribution on $\mathcal H$, then $PK$ denotes the induced distribution on
$\mathcal H'$:
\[
  (PK)(A)=\sum_{h\in\mathcal H} P(h)K(h,A).
\]
In words, to sample from $PK$, first draw $h\sim P$, and then draw the next
history according to $K(h,\cdot)$.

Suppose $K$ and $L$ are two such kernels and
\[
  \sup_{h\in\mathcal H}\TV\bigl(K(h,\cdot),L(h,\cdot)\bigr)\le \alpha .
\]
Then, for any two past-history laws $P,Q$,
\[
  \TV(PK,QL)\le \TV(P,Q)+\alpha .
\]
Indeed, by the triangle inequality,
\[
  \TV(PK,QL)
  \le
  \TV(PK,QK)+\TV(QK,QL).
\]
The first term is at most $\TV(P,Q)$ because applying the same randomized
transition cannot increase total variation. In finite spaces this follows from
\[
  \TV(PK,QK)
  =
  \frac12\sum_{z\in\mathcal H'}
  \left|
    \sum_{h\in\mathcal H} (P(h)-Q(h))K(h,z)
  \right|
  \le
  \frac12\sum_{h\in\mathcal H}|P(h)-Q(h)|
  =
  \TV(P,Q).
\]
For the second term,
\[
\begin{aligned}
  \TV(QK,QL)
  &=
  \frac12\sum_{z\in\mathcal H'}
  \left|
    \sum_{h\in\mathcal H} Q(h)\bigl(K(h,z)-L(h,z)\bigr)
  \right|  \\
  &\le
  \sum_{h\in\mathcal H} Q(h)\,
  \TV\bigl(K(h,\cdot),L(h,\cdot)\bigr)
  \le \alpha .
\end{aligned}
\]
This proves the one-step inequality.

We now identify the relevant kernels for the adaptive transcript. Fix a past
history $h$ before round $t$. If the learner has already stopped, then the next
history is obtained by appending the same deterministic dummy symbol under both
targets, so the two next-step kernels are identical.

Otherwise, the deterministic learner chooses a query set
\[
  S_t(h)\in\cS .
\]
Under target $p$, the next oracle response has law
$Q^\star_{p,S_t(h)}$; under target $q$, it has law
$Q^\star_{q,S_t(h)}$. The next transcript is obtained by appending the same
query $S_t(h)$ and the realized response. Equivalently, the next-step kernel
is the distribution obtained from the corresponding oracle response law by the
deterministic map
\[
  y\longmapsto h\circ(S_t(h),y).
\]
Total variation cannot increase under a deterministic map. Therefore
\[
  \TV\bigl(K_p^{(t)}(h,\cdot),K_q^{(t)}(h,\cdot)\bigr)
  \le
  \TV\bigl(Q^\star_{p,S_t(h)},Q^\star_{q,S_t(h)}\bigr)
  \le \alpha .
\]
Thus the one-step inequality applies with
\[
  P_t=P_{t-1}K_p^{(t)},
  \qquad
  Q_t=Q_{t-1}K_q^{(t)},
\]
and gives
\[
  \TV(P_t,Q_t)
  \le
  \TV(P_{t-1},Q_{t-1})+\alpha .
\]
By induction,
\[
  \TV(P_T,Q_T)\le T\alpha
\]
for every fixed seed $r$.

It remains to remove the conditioning on the learner's random seed. Let
$P_T^r$ and $Q_T^r$ be the conditional laws of the padded query-response
transcript $H_T$ given seed $R=r$, under targets $p$ and $q$ respectively.
Let $\nu$ be the common distribution of $R$.

First suppose that the seed is included in the full transcript. Let
$\mathsf P_T$ and $\mathsf Q_T$ denote the laws of the pair $(R,H_T)$ under
targets $p$ and $q$. For any event $E$ in the joint seed-transcript space,
define its section at seed $r$ by
\[
  E_r:=\{h:(r,h)\in E\}.
\]
Then
\[
  \mathsf P_T(E)=\int P_T^r(E_r)\,d\nu(r),
  \qquad
  \mathsf Q_T(E)=\int Q_T^r(E_r)\,d\nu(r).
\]
Hence
\[
\begin{aligned}
  |\mathsf P_T(E)-\mathsf Q_T(E)|
  &=
  \left|
    \int \bigl(P_T^r(E_r)-Q_T^r(E_r)\bigr)\,d\nu(r)
  \right|  \\
  &\le
  \int
  \left|P_T^r(E_r)-Q_T^r(E_r)\right|
  \,d\nu(r)  \\
  &\le
  \int \TV(P_T^r,Q_T^r)\,d\nu(r)
  \le T\alpha .
\end{aligned}
\]
Taking the supremum over all events $E$ gives
\[
  \TV(\mathsf P_T,\mathsf Q_T)\le T\alpha .
\]

If the seed is not included in the transcript, then the observed transcript is
obtained from the full pair $(R,H_T)$ by applying the deterministic projection
that discards $R$. Total variation cannot increase under deterministic
post-processing, so the same bound holds for the transcript law without the
seed as well. This proves the lemma.
\end{proof}

We also use the following standard result about the mixing time of Markov chains. We provide a proof for completeness.
\begin{lemma}[Doeblin minorization \citep{rosenthal1995minorization,roberts2004general}]
\label{lem:doeblin-minorization}
Let $K$ be a Markov kernel on a finite state space $\Omega$, and let $\pi$ be
a stationary distribution for $K$. Suppose that for some $\eta\in(0,1]$,
\[
  K(x,\cdot)\ge \eta\,\pi(\cdot)
  \qquad\text{for every }x\in\Omega .
\]
Then, for every initial distribution $\mu$ on $\Omega$ and every $t\ge 0$,
\[
  \TV(\mu K^t,\pi)\le (1-\eta)^t\TV(\mu,\pi)
  \le (1-\eta)^t .
\]
In particular, if $t\ge \eta^{-1}\log(1/\varepsilon),$
then $\TV(\mu K^t,\pi)\le \varepsilon$.
\end{lemma}

\begin{proof}
If $\eta=1$, then $K(x,\cdot)=\pi(\cdot)$ for every $x$, so the claim is
immediate. Assume $\eta<1$. Define
\[
  R(x,\cdot)
  :=
  \frac{K(x,\cdot)-\eta\pi(\cdot)}{1-\eta}.
\]
The minorization assumption implies that $R(x,\cdot)$ is a probability
distribution for every $x$, so $R$ is a Markov kernel. Thus
\[
  K=\eta\Pi_\pi+(1-\eta)R,
\]
where $\Pi_\pi$ is the rank-one kernel whose every row is $\pi$.

Since $\pi$ is stationary for $K$,
\[
  \pi
  =
  \pi K
  =
  \eta\pi+(1-\eta)\pi R,
\]
and hence $\pi R=\pi$. Therefore, for any distribution $\mu$,
\[
  \mu K-\pi
  =
  \eta\pi+(1-\eta)\mu R
  -
  \bigl(\eta\pi+(1-\eta)\pi R\bigr)
  =
  (1-\eta)(\mu R-\pi R).
\]
Taking total variation and using that Markov kernels contract total variation,
\[
  \TV(\mu K,\pi)
  \le
  (1-\eta)\TV(\mu,\pi).
\]
Iterating gives
\[
  \TV(\mu K^t,\pi)\le (1-\eta)^t\TV(\mu,\pi).
\]
Finally, $\TV(\mu,\pi)\le 1$ and
\[
  (1-\eta)^t\le e^{-\eta t},
\]
so $t\ge \eta^{-1}\log(1/\varepsilon)$ implies
$\TV(\mu K^t,\pi)\le\varepsilon$.
\end{proof}

\section{Examples of Hierarchically Comparable Query Families}
\label{sec:hc-examples}

We give several natural classes of query families satisfying
\Cref{def:HC}. The common theme is that, inside each cell of a balanced
partition tree, the family contains enough local queries to compare any two
points without leaving that cell.

\paragraph{Pairwise queries.}
The canonical example is the pairwise query family
\[
  \cS_{\mathrm{pair}}
  :=
  \bigl\{\{x,y\}:x,y\in U,\ x\neq y\bigr\}.
\]
This family is hierarchically comparable with respect to any balanced binary
partition tree over \(U\). Indeed, for every internal node \(C\) and every pair
\(x,y\in C\), the query set
\[
  W_C(\{x,y\}) := \{x,y\}
\]
belongs to \(\cS_{\mathrm{pair}}\) and satisfies
\[
  \{x,y\}\subseteq W_C(\{x,y\})\subseteq C.
\]
Thus pairwise access is the strongest and cleanest instance of local
comparability.

\paragraph{Hierarchical category menus.}
A less restrictive example arises when the domain is organized by a taxonomy or
category tree. Let \(\mathcal T\) be a balanced binary partition tree over
\(U\). For each internal category \(C\in\mathcal T\), suppose there is a local
menu family
\[
  \cS_C\subseteq 2^C\setminus\{\emptyset\}
\]
such that every pair of points in \(C\) appears together in some local menu:
for every \(x,y\in C\), there exists \(W\in\cS_C\) with
\[
  \{x,y\}\subseteq W\subseteq C.
\]
Then the union
\[
  \cS:=\bigcup_{C\in\mathcal T}\cS_C
\]
is hierarchically comparable on \(U\).

One concrete version is the following. Fix a constant \(K\ge 2\), and inside
each cell \(C\) allow all local menus of size \(\min\{K,|C|\}\):
\[
  \cS_C
  :=
  \bigl\{W\subseteq C: |W|=\min\{K,|C|\}\bigr\}.
\]
For any pair \(x,y\in C\), one can choose such a menu \(W\) containing
\(x,y\). When \(K=2\), this recovers pairwise queries; when \(K>2\), it gives
a model in which data providers expose small local menus rather than individual
pairs.

\paragraph{Intervals, boxes, and subcubes.}
Hierarchical comparability also holds for standard geometric query families.
If \(U=\{1,\ldots,n\}\) is ordered and \(\cS\) contains all intervals, take
\(\mathcal T\) to be the balanced dyadic interval tree. Each node \(C\) is an
interval. For any \(x,y\in C\), the interval
\[
  W_C(\{x,y\})
  :=
  [\min\{x,y\},\max\{x,y\}]
\]
is queryable, contains \(x,y\), and is contained in \(C\).

Similarly, if \(U\) is a finite grid and \(\cS\) contains all axis-aligned
boxes, take a balanced recursive spatial partition tree whose cells are boxes.
For any two points \(x,y\) in a box \(C\), the smallest axis-aligned box
containing \(x\) and \(y\) is contained in \(C\) and is queryable.

Finally, if \(U=\{0,1\}^d\) and \(\cS\) contains all subcubes, take a balanced
coordinate-splitting tree. Every node \(C\) is a subcube. For any
\(x,y\in C\), the smallest subcube containing \(x\) and \(y\) is contained in
\(C\) and is queryable. Hence the subcube family is hierarchically comparable.

These geometric examples should be interpreted as structural examples of the
condition. In cases where the full domain \(U\) itself is queryable, the
ordinary empirical learner already gives the near-linear rate; the value of
\Cref{def:HC} is that it also covers local-menu families where learning must
be assembled recursively from smaller comparisons

\section{Intermediate Exponents}\label{app:intermediate}

In this section we prove \Cref{thm:intermediate-rates}. The only difference from the equal-block
construction is that we allow the inflated blocks to have unequal sizes. This
removes the divisibility constraint \(n=1+(H-1)b\).

\begin{lemma}[Unequal block lift]
\label{thm:block-lift}
There exist universal constants \(H_0\ge 5\) and
\(c_0,\eps_0,C>0\) such that the following holds. Fix an odd integer
\[
  H=2m+1\ge H_0 .
\]
Let
\[
  A_1,B_1,\ldots,A_m,B_m
\]
be pairwise disjoint nonempty blocks, and set
\[
  U:=\{0\}\cup A_1\cup B_1\cup\cdots\cup A_m\cup B_m,
  \qquad
  n:=|U|,
\]
and
\[
\mathbf b
:=
\bigl(|A_1|,|B_1|,\ldots,|A_m|,|B_m|\bigr).
\]
There is a complete co-occurrence query family \(\cS_{H,\mathbf b}\) on \(U\)
such that
\[
q^\star_{\cS_{H,\mathbf b},\mathcal P_U}(\eps,1/3)
\ge
c_0\frac{n+H^2}{\eps^2}
  \qquad
  \text{for all }0<\eps\le \eps_0,
\]
while
\[
  q^\star_{\cS_{H,\mathbf b},\mathcal P_U}(\eps,\delta)
  \le
  C\,
  \frac{
    H(H+\log(1/\delta))\log(e/\eps)
    + n
    + H\log(eH/\delta)
  }{\eps^2}
  \qquad
  \text{for all }\eps,\delta\in(0,1).
\]
\end{lemma}

\begin{proof}
Let the coarse domain be
\[
  \overline U
  =
  \{0\}
  \cup
  \{\overline a_1,\overline b_1,\ldots,\overline a_m,\overline b_m\}.
\]
The point \(0\) is the hub. The coarse atoms
\(\overline a_j,\overline b_j\) are inflated into the blocks \(A_j,B_j\).

Define the block-query family
\[
  \mathcal B
  :=
  \{A_i:1\le i\le m\}
  \cup
  \{B_i:1\le i\le m\},
\]
and the lifted hard-query family
\[
  \mathcal T
  :=
  \{T_{jk}:1\le j<k\le m\},
  \qquad
  T_{jk}:=\{0\}\cup A_j\cup B_j\cup A_k\cup B_k .
\]
Set
\[
  \cS_{H,\mathbf b}:=\mathcal B\cup\mathcal T .
\]
The co-occurrence graph is complete. Points in the same block co-occur in that
block query. Points in different blocks co-occur in some \(T_{jk}\). The hub
co-occurs with every non-hub point in some \(T_{jk}\).

\paragraph{Lower bound.}
Let
\[
  \overline{\cS}_H
  :=
  \{\overline T_{jk}:1\le j<k\le m\},
  \qquad
  \overline T_{jk}
  :=
  \{0,\overline a_j,\overline b_j,\overline a_k,\overline b_k\}.
\]
This is the coarse hard family from \Cref{thm:quadratic-lower-main}. The quadratic
lower bound gives a full-support hard subfamily
\[
  \overline{\mathcal P}_{H}\subseteq \Delta_+(\overline U)
\]
such that every learner for the coarse problem requires at least
\[
  c_0\frac{H^2}{\eps^2}
\]
queries for all \(0<\eps\le\eps_0\), after adjusting universal constants if
necessary.

For a coarse distribution \(\overline p\in\Delta(\overline U)\), define its
lift to \(U\) by
\[
  \mathrm{Lift}(\overline p)(0):=\overline p(0),
\]
\[
  \mathrm{Lift}(\overline p)(x)
  :=
  \frac{\overline p(\overline a_j)}{|A_j|}
  \quad (x\in A_j),
  \qquad
  \mathrm{Lift}(\overline p)(x)
  :=
  \frac{\overline p(\overline b_j)}{|B_j|}
  \quad (x\in B_j).
\]
Thus a lifted hard distribution has the same coarse block masses as
\(\overline p\) and is uniform inside each block. Since
\(\overline{\mathcal P}_H\) is full support, every lifted hard distribution has
exact support \(U\).

Suppose there were a learner \(L\) for the lifted problem using \(q\) queries.
We construct a coarse learner \(\overline L\) with no larger query count by
simulating the transcript of \(L\). If \(L\) queries a block \(A_j\), then
\(\overline L\) uses its own randomness to return a uniformly random point of
\(A_j\). This is the correct conditional law under every lifted hard target,
because the lifted target is uniform inside \(A_j\). The same applies to
queries of \(B_j\).

If \(L\) queries \(T_{jk}\), then \(\overline L\) queries
\(\overline T_{jk}\). If the coarse oracle returns \(0\), the simulator returns
\(0\). If it returns \(\overline a_j\), the simulator returns a uniformly random
point of \(A_j\), and similarly for the other three non-hub coarse atoms. For
example, for \(x\in A_j\),
\[
  \Pr[\text{simulator returns }x]
  =
  \frac{\overline p(\overline a_j)}
       {\overline p(\overline T_{jk})}
  \cdot \frac1{|A_j|}
  =
  \frac{\mathrm{Lift}(\overline p)(x)}
       {\mathrm{Lift}(\overline p)(T_{jk})}.
\]
Thus the simulated response law is exactly the lifted conditional response law.
By induction over adaptive rounds, the entire simulated transcript has exactly
the same law as the transcript of \(L\) on the lifted target.

When \(L\) outputs \(\widehat p\) on \(U\), the coarse learner outputs the
block projection \(\Pi\widehat p\), defined by
\[
  (\Pi\widehat p)(0)=\widehat p(0),
  \qquad
  (\Pi\widehat p)(\overline a_j)=\widehat p(A_j),
  \qquad
  (\Pi\widehat p)(\overline b_j)=\widehat p(B_j).
\]
For every coarse target \(\overline p\),
\[
  \Pi\,\mathrm{Lift}(\overline p)=\overline p.
\]
Moreover, total variation contracts under this deterministic coarse-graining:
\[
  \TV\bigl(\Pi\widehat p,\overline p\bigr)
  =
  \TV\bigl(\Pi\widehat p,\Pi\,\mathrm{Lift}(\overline p)\bigr)
  \le
  \TV\bigl(\widehat p,\mathrm{Lift}(\overline p)\bigr).
\]
Hence any successful lifted learner would give a successful coarse learner with
no larger query count. The coarse lower bound transfers:
\[
  q^\star_{\cS_{H,\mathbf b},\mathcal P_U}(\eps,1/3)
  \ge
  c_0\frac{H^2}{\eps^2}.
\]

Independently, \Cref{thm:full-cond-main} implies
\[
q^\star_{\cS_{H,\mathbf b},\mathcal P_U}(\eps,1/3)
\ge c\,\frac{n}{\eps^2},
\]
because unrestricted conditional access is at least as informative as access
to $\cS_{H,\mathbf b}$, and $\mathcal P_U$ contains
$\mathcal P_U^+$. Combining this with the transferred coarse lower bound,
\[
q^\star_{\cS_{H,\mathbf b},\mathcal P_U}(\eps,1/3)
\ge c'\frac{n+H^2}{\eps^2}.
\]

\paragraph{Upper bound.}
Fix $p\in\mathcal P_U$. Define the block-mass vector $w$ as above. For each
block
$
D\in\{A_1,B_1,\ldots,A_m,B_m\},
$
define
\[
p_D :=
\begin{cases}
p(\cdot\mid D), & p(D)>0,\\
u_D, & p(D)=0.
\end{cases}
\]

The learner has two stages. First, using the lifted hard queries \(T_{jk}\) and
projecting each response to its coarse atom, the learner obtains conditional
samples from \(w(\cdot\mid \overline T_{jk})\). Since the coarse co-occurrence
graph is complete, \Cref{thm:generic-upper-main} learns \(w\) to TV error
\(\eps/4\) and failure probability \(\delta/2\) using
\[
  O\!\left(
    \frac{
      H(H+\log(1/\delta))\log(e/\eps)
    }{\eps^2}
  \right)
\]
queries.

Second, the learner estimates the internal conditional distribution inside
each block. Since each block is queryable, \Cref{lem:empirical} and a union
bound over the \(2m\) blocks give simultaneous TV error at most \(\eps/4\) in
every block with failure probability at most \(\delta/2\), using
\[
  O\!\left(
    \frac{
      \sum_{j=1}^m |A_j|+\sum_{j=1}^m |B_j|
      + H\log(eH/\delta)
    }{\eps^2}
  \right)
  =
  O\!\left(
    \frac{
      n+H\log(eH/\delta)
    }{\eps^2}
  \right)
\]
queries.

Let \(\widehat w\) be the block-mass estimate and \(\widehat p_D\) the
conditional estimate inside block \(D\). Define
\[
  \widehat p
  :=
  \widehat w_0\delta_0
  +
  \sum_{j=1}^m
  \widehat w_{\overline a_j}\widehat p_{A_j}
  +
  \sum_{j=1}^m
  \widehat w_{\overline b_j}\widehat p_{B_j}.
\]
Also define
\[
  \widetilde p
  :=
  w_0\delta_0
  +
  \sum_{j=1}^m
  w_{\overline a_j}\widehat p_{A_j}
  +
  \sum_{j=1}^m
  w_{\overline b_j}\widehat p_{B_j}.
\]
On the good event for the block estimates,
\[
  \TV(p,\widetilde p)
  \le
  \sum_{j=1}^m
  w_{\overline a_j}\TV(p_{A_j},\widehat p_{A_j})
  +
  \sum_{j=1}^m
  w_{\overline b_j}\TV(p_{B_j},\widehat p_{B_j})
  \le \frac{\eps}{4}.
\]
On the good event for the block-mass estimate,
\[
  \TV(\widetilde p,\widehat p)=\TV(w,\widehat w)\le \frac{\eps}{4}.
\]
Therefore
\[
  \TV(p,\widehat p)\le \frac{\eps}{2}<\eps.
\]
The total query bound is
\[
  O\!\left(
  \frac{
    H(H+\log(1/\delta))\log(e/\eps)
    + n
    + H\log(eH/\delta)
  }{\eps^2}
  \right).
\]
This proves the theorem.
\end{proof}
\begin{proof}[Proof of \Cref{thm:intermediate-rates}]
Fix \(\gamma\in(0,1)\), and set
\[
  a:=\frac{1+\gamma}{2}.
\]
For every sufficiently large \(n\), let \(H_n\) be the largest odd integer at
most \(n^a\). Write
\[
  H_n=2m_n+1.
\]
By increasing \(N_\gamma\) if necessary, we may assume that for all
\(n\ge N_\gamma\),
\[
  H_n\ge H_0,
  \qquad
  H_n\le n,
  \qquad
  H_n\ge \frac12 n^a .
\]
The last inequality holds for all large \(n\) because the largest odd integer
below \(n^a\) is at least \(n^a-2\).

Since \(H_n\le n\), we have
\[
  2m_n=H_n-1\le n-1.
\]
Thus we can partition the \(n-1\) non-hub elements of \([n]\) into
\(2m_n\) nonempty blocks
\[
  A_1,B_1,\ldots,A_{m_n},B_{m_n}.
\]
Apply \Cref{thm:block-lift} with \(H=H_n\) and these blocks, and let
\(\sellers=\cS_{H_n,\mathbf b}\). By construction, the co-occurrence graph of
\(\sellers\) on \([n]\) is complete.

For the lower bound, \Cref{thm:block-lift} gives
\[
  q^\star_{\sellers,\mathcal P_{[n]}}(\eps,1/3)
  \ge
  c_0\frac{H_n^2}{\eps^2}.
\]
Since \(H_n\ge \frac12 n^a\) and \(2a=1+\gamma\),
\[
  H_n^2\ge \frac14 n^{1+\gamma}.
\]
Therefore
\[
  q^\star_{\sellers,\mathcal P_{[n]}}(\eps,1/3)
  \ge
  \frac{c_0}{4}\frac{n^{1+\gamma}}{\eps^2}
\]
for all \(0<\eps\le\eps_0\). Set
\[
  c_1:=\frac{c_0}{4}.
\]

For the upper bound, apply \Cref{thm:block-lift} with \(\delta=1/3\):
\[
  q^\star_{\sellers,\mathcal P_{[n]}}(\eps,1/3)
  \le
  C\,
  \frac{
    H_n(H_n+1)\log(e/\eps)
    + n
    + H_n\log(eH_n)
  }{\eps^2},
\]
where we absorbed constants depending on \(\delta=1/3\).

Using
\[
  H_n^2\le n^{1+\gamma},
  \qquad
  H_n\le n,
  \qquad
  n\le n^{1+\gamma},
\]
and \(n\ge 3\), the numerator is at most
\[
  C'\,n^{1+\gamma}\log(e/\eps)\log n
\]
for a universal constant \(C'\). Finally, shrink \(\eps_0\) if necessary so
that \(\eps_0\le e^{-1}\). Then for \(0<\eps\le\eps_0\),
\[
  \log(e/\eps)\le 2\log(1/\eps).
\]
Thus
\[
  q^\star_{\sellers,\mathcal P_{[n]}}(\eps,1/3)
  \le
  c_2\,
  \frac{
    n^{1+\gamma}\log(1/\eps)\log n
  }{\eps^2}
\]
for a universal constant \(c_2>0\).

Combining the lower and upper bounds proves the theorem.
\end{proof}

\end{document}